\documentclass[11pt]{article}

\newif\ifcomments
\commentstrue

\usepackage{fullpage}
\usepackage[utf8]{inputenc}

\usepackage[linktocpage=true,
pagebackref=true,colorlinks,
linkcolor=DarkRed,citecolor=ForestGreen,
bookmarks,bookmarksopen,bookmarksnumbered]
{hyperref}

\usepackage{nicefrac}
\usepackage[normalem]{ulem}
\usepackage[letterpaper, margin=1in]{geometry}
\usepackage{amsmath,amsthm,amssymb,amsfonts}
\usepackage{fancyhdr}
\usepackage{enumitem}
\usepackage{graphicx}
\usepackage{xspace}
\usepackage{subfigure}
\usepackage{mdframed}
\usepackage[usenames,dvipsnames]{xcolor}
\usepackage{tcolorbox}
\usepackage[]{todonotes}
\usepackage[linesnumbered,ruled,vlined]{algorithm2e}
\usepackage{tikz}
\tcbuselibrary{skins,breakable}
\tcbset{enhanced jigsaw}

\usepackage{thmtools}
\usepackage{thm-restate}

\definecolor{DarkRed}{rgb}{0.5,0.1,0.1}
\definecolor{ForestGreen}{rgb}{0.1333,0.5451,0.1333}
\definecolor{Red}{rgb}{1,0,0}

\definecolor{ForestGreen}{rgb}{0.1333,0.5451,0.1333}
\usepackage[linktocpage=true,
	pagebackref=true,colorlinks,
	linkcolor=DarkRed,citecolor=ForestGreen,
	bookmarks,bookmarksopen,bookmarksnumbered]
	{hyperref}
\usepackage[noabbrev,nameinlink]{cleveref}
\crefname{property}{property}{Property}
\creflabelformat{property}{(#1)#2#3}
\crefname{equation}{eq}{Eq}
\creflabelformat{equation}{(#1)#2#3}

\newtheorem{thm}{Theorem}[section]
\newtheorem{theorem}[thm]{Theorem}
\newtheorem{cor}[thm]{Corollary}

\newtheorem{prop}[thm]{Proposition}

\newtheorem{conjecture}[thm]{Conjecture}

\newtheorem{lem}[thm]{Lemma}
\newtheorem{lemma}[thm]{Lemma}

\newtheorem{definition}[thm]{Definition}
\newtheorem{obs}[thm]{Observation}
\newtheorem{claim}[thm]{Claim}

\newtheorem{problem}[thm]{Problem}
\newtheorem{remark}[thm]{Remark}

\newtheorem*{solution*}{Solution}

\DeclareMathOperator*{\Prob}{\ensuremath{\textnormal{Pr}}}
\renewcommand{\Pr}{\Prob}

\newcommand{\poly}{\mbox{\rm poly}}
\newcommand{\eps}{\varepsilon}
\newcommand{\MST}{\mathrm{MST}}

\newcommand{\nmp}{\textsc{RecMincut}\xspace}
\newcommand{\lp}{\textsc{LocalPrim}\xspace}
\newcommand{\blp}{\textsc{LocalKCut}\xspace}
\newcommand{\mcut}{\textsc{Mincut}\xspace}
\newcommand{\mcp}{\textsc{MaximalPartition}\xspace}
\newcommand{\kcp}{\textsc{kCutPartition}\xspace}
\newcommand{\ins}{\textsc{Insert}\xspace}
\newcommand{\del}{\textsc{Delete}\xspace}
\newcommand{\nextedge}{\textsc{NextEdge}\xspace}

\newcommand{\localmst}{\textsc{LocalMST}\xspace}
\newcommand{\minextset}{\textsc{LocalKCut}\xspace}

\newcommand{\cP}{\mathcal{P}}
\newcommand{\Otil}{\tilde{O}}
\newcommand{\vol}{\mathrm{vol}}

\global\long\def\mincut{\mathrm{mincut}}
\global\long\def\Ohat{\widehat{O}}
\global\long\def\poly{\mathrm{poly}}
\global\long\def\polylog{\mathrm{polylog}}

\AtBeginDocument{%
  \DeclareFontShape{OT1}{cmr}{m}{scit}{<->ssub*cmr/m/sc}{}%
  \DeclareFontShape{OT1}{cmr}{bx}{sc}{<->ssub*cmr/m/n}{}
}

\ifcomments

\def\thatchaphol#1{\marginpar{$\leftarrow$\fbox{TS}}\footnote{$\Rightarrow$~{\sf\textcolor{purple}{#1 --Thatchaphol}}}}
\def\chaitanya#1{\marginpar{$\leftarrow$\fbox{CN}}\footnote{$\Rightarrow$~{\sf\textcolor{blue}{#1 --Chaitanya}}}}
\newcommand{\tnote}[1]{\colnote{red}{#1--Ohad}{TS}}
\newcommand{\ynote}[1]{\colnote{olive}{#1--Amir}{YL}}

\else

\newcommand{\tnote}[1]{}
\newcommand{\ynote}[1]{}
\newcommand{\thatchaphol}[1]{}
\newcommand{\chaitanya}[1]{}
\fi %

\title{Maximal $k$-Edge-Connected Subgraphs in Weighted Graphs\\via Local Random
Contraction}
\author{Chaitanya Nalam\thanks{University of Michigan, Ann Arbor. Email: {\tt nalamsai@umich.edu}} \and Thatchaphol Saranurak\thanks{University of Michigan, Ann Arbor. Email: {\tt thsa@umich.edu}}}
\begin{document}

\maketitle

\begin{abstract}
The \emph{maximal $k$-edge-connected subgraphs} problem is a classical graph clustering problem studied since the 70's. Surprisingly, no non-trivial technique for this problem
in weighted graphs is known: a very straightforward recursive-mincut algorithm with $\Omega(mn)$ time has remained the fastest algorithm until now.
All previous progress gives a speed-up only when the graph is unweighted, and $k$ is small enough (e.g.~Henzinger~et~al.~(ICALP'15), Chechik~et~al.~(SODA'17), and Forster~et~al.~(SODA'20)).

We give the first algorithm that breaks through the long-standing
$\Otil(mn)$-time barrier in \emph{weighted undirected} graphs. More specifically, we show a maximal $k$-edge-connected subgraphs algorithm
that takes only $\tilde{O}(m\cdot\min\{m^{3/4},n^{4/5}\})$ time. As an immediate application, we
can $(1+\epsilon)$-approximate the \emph{strength} of all edges in
undirected graphs in the same running time.

Our key technique is the first local cut algorithm with \emph{exact}
cut-value guarantees whose running time depends only on the output
size. All previous local cut algorithms either have running time depending
on the cut value of the output, which can be arbitrarily slow in weighted graphs or have approximate cut guarantees.
\end{abstract}

\newpage
\tableofcontents
\newpage

\section{Introduction}

We revisit a natural graph clustering problem of efficiently computing
\emph{maximal $k$-edge-connected subgraphs} that has been studied since
the 70's \cite{tarjan1972depth}. Given a (directed) graph $G=(V,E)$,
we say that $G$ is \emph{$k$-edge-connected} if all the edge cuts have
value at least $k$ and, similarly, $G$ is \emph{$k$-vertex-connected
} if all vertex cuts have value at least $k$. The \emph{maximal $k$-edge-connected
subgraphs} of $G$ contains all subsets $\{V_{1},\dots,V_{z}\}$ of
$V$ such that, for each $i$, the induced subgraph $G[V_{i}]$ is $k$-edge-connected
and there is no strict superset $V'_{i}\supset V_{i}$ where $G[V'_{i}]$
is $k$-edge-connected. The \emph{maximal $k$-vertex-connected subgraphs}
of $G$ is defined analogously. 

It turns out that all the variants of these problems (whether the graph is weighted/unweighted or directed/undirected) admit a very
simple \emph{recursive-mincut} algorithm: For the $k$-edge-connectivity
version, compute an edge cut $(A, B)$ of value less than $k$
by calling a global minimum cut subroutine and return $\{V\}$ if
no such cut exists. Otherwise, recurse on both $G[A]$ and $G[B]$
and return the union of answers of the two recursions. For the $k$-vertex-connectivity
version, similarly compute a vertex cut $(A, S, B)$ of value less than
$k$ (where $S$ separates $A$ from $B$, i.e., $E(A,B)=\emptyset$).
If such cut exists, recurse on both $G[A\cup S]$ and $G[B\cup S]$.
The worst-case running time of this algorithm on a graph with $n$
vertices and $m$ edges is $\Theta(n\cdot T_{\mincut})=\Omega(mn)$
where $T_{\mincut}=\Omega(m)$ is the time for computing minimum cuts.\footnote{Karger showed that \cite{karger2000minimum} it takes $T_{\mincut}=\Otil(m)$
to compute minimum edge cuts in undirected graphs, but the start-of-the-art
for other variants are worse. For minimum edge cuts in directed graphs,
$T_{\mincut}=\Ohat(m\cdot\min\{n/m^{1/3},\sqrt{n}\})$ \cite{cen2022minimum}
and $T_{\mincut}=\Otil(mk)$ if the graph is unweighted \cite{gabow1995matroid}.
For minimum vertex cuts, $T_{\mincut}=\Otil(mn)$ \cite{henzinger2000computing}
in both undirected and directed graphs and, if the graph is unweighted,
we have $T_{\mincut}=\Ohat(m)$ for undirected graphs \cite{li2021vertex}
and $T_{\mincut}=\min\{\Ohat(n^{2}),\Otil(mk^{2})\}$ for directed
graphs \cite{li2021vertex,forster2020computing}.}\footnote{Throughout the paper, we use $\Otil(\cdot)$ to hide $\polylog(n)$ factor and use $\Ohat(\cdot)$ to hide $n^{o(1)}$ factor.}

Surprisingly, the above straightforward algorithm has remained the
fastest algorithm! Although the problem has been studied
since the 70's, all known algorithms give a speed-up only when the
graph is unweighted and $k$ is small enough. The classical depth-first-search
algorithm \cite{tarjan1972depth} solves both maximal $2$-edge-connected
and $2$-vertex-connected subgraphs problems in undirected unweighted
graphs in $O(m)$ time. There was an effort to generalize this result
to directed graphs \cite{erusalimskii1980bijoin,makino1988algorithm,jaberi2016computing},
but the algorithms still take $\Omega(mn)$ time. Henzinger~et~al.~\cite{henzinger2015finding}
later obtained $\Otil(n^{2})$-time algorithms in directed graphs
when $k=2$. Then, Chechik~et~al.~\cite{chechik2017faster} gave the first
algorithm that breaks the $O(mn)$ bound for any constant $k$. 
Their algorithm takes $\Otil(m\sqrt{m}k^{O(k)})$ time for all variants of the problems. They also gave an improved bound of $\Otil(m\sqrt{n}k^{O(k)})$ for computing maximal $k$-edge-connected subgraphs in undirected graphs. Both bounds were then improved by  Forster~et~al.~\cite{forster2020computing} to 
$\Otil(m\sqrt{m}k^{3/2})$ and $\Otil(mk+n^{3/2}k^{3})$, respectively. Nonetheless, in weighted graphs where $k$ can be arbitrarily large or even in unweighted
graphs when $k\ge m^{1/3}/n^{1/6}$, the simple recursive-mincut algorithm remains the fastest algorithm until now.

\paragraph{Our Results.}

We give the first algorithm that breaks through the long-standing
$\Otil(mn)$-time barrier of the recursive-mincut algorithm in \emph{weighted}
graphs, in the case of maximal \emph{$k$-edge-connected} subgraphs
in \emph{undirected} graphs, which is the variant of the problem that
receives attention in practice and several heuristics were developed
\cite{zhou2012finding,akiba2013linear}. Formally, our main result
is as follows:
\begin{thm}
\label{thm:main}There is a randomized algorithm that, given a weighted
undirected graph $G=(V,E,w)$ where $w\in\mathbb{Z}_{\ge0}^{E}$ with
$n$ vertices and $m$ edges and any parameter $k>0$, computes the
maximal $k$-edge-connected subgraphs of $G$ w.h.p.~in $O(m\cdot\min\{m^{3/4}\log^{3.75}n,n^{4/5}\log^{3.6}n\})$
time.
\end{thm}

It is known that the solution $\{V_{1},\dots,V_{z}\}$, the maximal
$k$-edge-connected subgraphs of $G$, is unique and forms a partition
of $V$. Moreover, it precisely determines whether an edge has \emph{strength}
at least $k$: each edge $(u,v)\in E$ has strength at least $k$
iff both $u$ and $v$ are inside the same subset $V_{i}\in\{V_{1},\dots,V_{z}\}$.
The edge strength introduced by Benczur and Karger \cite{benczur2002randomized}
is a central notion for graph sparsification and its generalization
\cite{fung2019general,chen2020near,cen2021sparsification}. \cite{benczur2002randomized} gave a near-linear algorithm that underestimates the edge strength. 
Although their estimates are sufficient for graph sparsification \cite{benczur2002randomized,chekuri2018minimum,fung2019general}, these estimates do not give an approximation for edge-strength of \emph{every} edge; the estimate of some edge can deviate from its strength by even a $\Omega(n^{1/3})$ factor. See \Cref{counterexample} for the discussion.

The fastest known algorithm for computing exact edge strengths is again the recursive-mincut algorithm which requires
$\Otil(mn)$ time. No faster algorithm is known, even when a large approximation ratio is allowed. In contrast, \Cref{thm:main} 
improves the state-of-the-art
for approximating edge strength by giving an $(1+\eps)$-approximation.

\begin{cor}\label{cor:apprxedgstn}
There is a randomized algorithm that, given a weighted undirected
graph, can w.h.p.~$(1+\epsilon)$-approximates the strength of all
edges in $\Otil(m\cdot\min\{m^{3/4},n^{4/5}\}/\epsilon)$ time.
\end{cor}

\paragraph{Main Technical Contribution:  Local Exact Algorithms in Weighted Graphs.}
The critical tool behind \Cref{thm:main} is the first \emph{local cut algorithm}
that has an exact cut-value guarantee and yet is fast in weighted graphs. 

We generally refer to local cut algorithms as the algorithms that,
given a seed vertex $x$, find a cut $S$ containing $x$ with desirable
properties, but only spend time close to the size of $S$ and independent
of the size of the whole graph. For example, an influential line of
work on local cut algorithms for finding sparse cuts (e.g.~\cite{spielman2013local,andersen2006local,orecchia2014flow,henzinger2020local})
found wide-range applications, including solving linear systems \cite{spielman2014nearly}
and dynamic graph algorithms \cite{nanongkai2017dynamic,saranurak2019expander}.
Unfortunately, all these algorithms only have \emph{approximate} guarantees.\footnote{If there is a cut $S$ with conductance $\phi$, these algorithms
find a cut $S'$ with conductance $\phi'$ where $\phi'>\phi$.} There is another class of algorithms initiated in the area of property
testing that checks whether there is a small cut $S$ containing the
seed $x$ whose cut size is at most $k$. These algorithms have running
time $\exp(k)$ \cite{yoshida2010testing,orenstein2011testing,yoshida2012property}
and can be improved to $\poly(k)$ \cite{goldreich1997property,parnas2002testing,forster2020computing,nanongkai2019breaking}.
Although these algorithms found interesting applications beyond property
testing \cite{forster2020computing,chalermsook2021vertex,jiang2022vertex},
they are not useful in weighted graphs because $k$ can be arbitrarily
large. 

To summarize, all local cut algorithms in the literature are either
approximate or run in time proportional to the cut value $k$. Our
key technique is an algorithm that bypasses both these drawbacks.

To state our key tool, we need some notations. From
now, all graphs $G=(V,E,w)$ we consider are undirected. For any edge
set $F\subseteq E$, $w(F)=\sum_{e\in F}w(e)$ is the total weight
of $F$. For any set $X\subset V$, we write $\overline{X}=V\setminus X$.

Let $\delta(X)=E(X,\overline{X})$ denote the \emph{cut set} of $X$
and $w(\delta(X))$ is denote the \emph{cut value} of $X$. The (unweighted)
\emph{volume} of $X$ denoted by $\vol(X)=|E(X,V)|$ counts all edges
incident to $X$. We say that $X$ is an \emph{extreme set} if, for
every strict subset $Y\subset X$, $w(\delta(Y))>w(\delta(X))$. Equivalently,
$X$ is the unique minimum cut in $G/\overline{X}$, the graph after
contracting $\overline{X}$ into a single vertex.

\begin{definition}
For any vertex $x\in V$, we say $X$ is a \emph{$(x,\nu,\sigma,k)$-set} if $x\in X$, $\vol(X)<\nu$, $|X|<\sigma$, and $w(\delta(X))<k$. If $X$ is also an extreme set, then $X$ is called an \emph{$(x,\nu,\sigma,k)$-extreme set}.
\end{definition}

Intuitively, an $(x,\nu,\sigma,k)$-set is a ``small local cut''
containing the \emph{seed} vertex $x$ that we want to find. From now on we use $\nu$ to represent the volume of the small local cut and $\sigma$ to represent the number of vertices on the smaller side of the local cut. 
The key problem we consider is the following:
\begin{restatable}[\minextset]{problem}{localkcutprob}
\label{prob:minextset}
Given a graph $G=(V,E,w)$, a vertex $x$, and parameters $\nu,\sigma$ and $k$, 
either find an $(x,\nu,\sigma,k)$-set or return $\bot$ indicating that no $(x,\nu,\sigma,k)$-extreme set exists.
\end{restatable} 
Note that when there is no $(x,\nu,\sigma,k)$-extreme set but there exists an $(x,\nu,\sigma,k)$-set, then we allow returning either $\bot$ or any $(x,\nu,\sigma,k)$-set.
Now, our key technical contribution is as follows:
\begin{thm}
\label{thm:local main}There are randomized algorithms with the following
guarantees with high probability:
\begin{enumerate}
\item \label{enu:local vol}Given access to the adjacency lists $G$, \Cref{prob:minextset}
can be solved in time $O(\nu\sigma^{2}\log^{2}n)$.
\item \label{enu:local size}
After $O(m\sigma^{2}\log^{2}n)$ preprocessing
time on $G$, \Cref{prob:minextset} can be solved in time
$O(\sigma^{4}\log n)$.
\end{enumerate}
\end{thm}

The most important guarantee of \Cref{thm:local main} is that its
running time is independent from $k$. \Cref{thm:local main} has an
exact cut-size guarantee in the sense that there is no gap in $k$:
it either returns a $(x,\nu,\sigma,k)$-\textbf{set} or reports that no $(x,\nu,\sigma,k)$-extreme set exists. It is possible to strengthen the guarantee of \Cref{thm:local main} to ``either return an $(x,\nu,\sigma,k)$-\textbf{extreme set} or report that no $(x,\nu,\sigma,k)$-extreme set exists''. We show such an algorithm in \Cref{sec:minextset}. However, this algorithm is slower and so we will not use it for proving \Cref{thm:main}. Hence, we defer it to the appendix. 

We complement this algorithm with conditional lower bounds:

\begin{thm}\label{thm:lb main}
Given the parameters $x,\nu,\sigma,k$, checking whether an $(x,\nu,\sigma,k)$-set
exists is W[1]-hard even when parameterized by $\sigma$. Moreover, assuming
the OMv conjecture \cite{henzinger2015unifying},
there is no algorithm that, given access to the adjacency lists and matrix of $G$ (with no further preprocessing), solves  \Cref{prob:minextset} in $O(\sigma^{2-\Omega(1)})$ time.
\end{thm}

The first statement in \Cref{thm:lb main} notes that the notion of extreme sets in \Cref{prob:minextset} is crucial for efficient algorithm; if we just want to check if an $(x,\nu,\sigma,k)$-set exists, then it is unlikely that there exists an algorithm with $f(\sigma)n^{O(1)}$ time where $f$ is an arbitrary function (see the definition of W[1]-hardness from Section 13.3 of \cite{parametrizedCyganFKLMPPS15}).

The second statement of \Cref{thm:lb main} shows that the running time of \Cref{thm:local main} cannot be improved beyond $O(\sigma^{2-o(1)})$ if we only have access to the adjacency lists and matrix. Improving the running time of \Cref{thm:local main} to $\Otil(\sigma^{2})$, or even $\Otil(\sigma)$ with non-trivial pre-processing, would speed up the running time of \Cref{thm:main} to $\Otil(mn^{2/3})$ and $\Otil(mn^{1/2})$, respectively. We leave this as a very interesting open problem.

\paragraph{Key Technique: Local Random Contraction.}
The high-level message behind \Cref{thm:local main} is as follows: the famous \emph{random contraction} technique by \cite{karger1993global} for computing global minimum cut can be localized into local cut algorithms for detecting unbalanced extreme sets in \emph{weighted graphs}.

Applications of the random contraction technique to local cut algorithms appeared earlier in the property testing algorithms by Goldreich and Ron \cite{goldreich1997property} and also in \cite{parnas2002testing}. However, it is unclear a priori that their results extend to weighted graphs. 
This is because in setting of \cite{goldreich1997property,parnas2002testing} they have that (1) the graph is unweighted and has bounded degree, (2) the target cut $X$ they want to find is such that $|\delta(X)|< k$ and $|X| = O(k)$, and (3) their stated running time was $\poly(k)$. So the time bound inherently may depend on $k$. Some parts of the argument in \cite{goldreich1997property} crucially exploit that the graph is unweighted and does not carry on to our weighted setting. 
Our observation is that the ``local random contraction'' technique still works in weighted graphs with the running time bound independent from $k$.

This shows an exciting aspect of the random contraction technique because, to our best knowledge, there is an inherent dependency on $k$ in all other techniques in the current literature for local cut algorithms with exact cut-value guarantees. These techniques include local-search-based algorithms with deterministic branching strategies \cite{yoshida2010testing,orenstein2011testing,chechik2017faster}, local-search-based algorithm with randomized stopping points \cite{forster2020computing}, and even the localized version of the Goldberg-Rao flow algorithm \cite{orecchia2014flow,nanongkai2019breaking}. 

We view \Cref{thm:local main} as the first step toward stronger local
cut algorithms in weighted graphs with exact cut-value guarantees.
Because of the strong guarantee of these algorithms, we expect many
more applications. For example, strengthening \Cref{thm:local main}
to work in directed graphs would immediately generalize our main result,
\Cref{thm:main}, to work in directed graphs.

\paragraph{Related Problems. }

The related problem of computing \emph{$k$-edge-connected components
}of a graph $G=(V,E)$ is to find the unique partition $\{V'_{1},\dots,V'_{z'}\}$
of $V$ such that, for every vertex pair $(s,t)$ in the same part
$V'_{i}$, the $(s,t)$-minimum cut in $G$ (not in $G[V_{i}']$)
is at least $k$. The maximal $k$-edge-connected subgraphs $\{V_{1},\dots,V_{z}\}$
can be different from the $k$-edge-connected components $\{V'_{1},\dots,V'_{z}\}$.
Consider for example an unweighted undirected graph $G'$ where vertices $s$
and $t$ are connected through three parallel paths $(s,u_{1},t)$,
$(s,u_{2},t)$, and $(s,u_{3},t)$. While $s$ and $t$ are in the
same $3$-edge-connected component of $G'$, no two vertices are in
the same maximal $3$-edge-connected subgraphs of $G'$. 

We note that computing $k$-edge-connected components in an undirected graph $G$ can be reduced
to computing all-pairs max flow in $G$, which can be done in $\Otil(n^{2})$
time using the recent algorithm of \cite{abboud2021gomory}. This
algorithm, however, does not solve nor imply anything to our problem of computing maximal $k$-edge-connected subgraphs.

\paragraph{Organization. }
After giving a high-level technical overview in \Cref{sec:overview},
we start by formally describing the recursive mincut algorithm for finding the maximal $k$-edge connected subgraphs in \Cref{sec:prelims} as a warm-up for similar proofs. In \Cref{sec:extsets}, we show a local cut algorithm for \Cref{thm:local main}(\ref{enu:local vol}). We then use this local cut algorithm to obtain the maximal $k$-edge-connected subgraphs in $\Otil(mm^{3/4})$ time in \Cref{sec:recdepth}. This proves the first running time claimed in \Cref{thm:main}.
In \Cref{sec:localalg2}, we give another local cut algorithm stated in \Cref{thm:local main}(\ref{enu:local size}) which is faster than the one in \Cref{sec:extsets} when the volume of the cut is high. Similar to \Cref{sec:recdepth}, in \Cref{sec:paramchange}, we apply this local algorithm to obtain the maximal $k$-edge-connected subgraphs in $\Otil(mn^{4/5})$ time, completing the proof of \Cref{thm:main}.

In \Cref{sec:lb} we show hardness for solving \Cref{prob:minextset}, proving the \Cref{thm:lb main}. In \Cref{sec:apprxes}, we show an application of our main result to approximating edge strength, i.e.,   \Cref{cor:apprxedgstn}.

\section{Overview of techniques}
\label{sec:overview}
We say that $X\subset V$ is a $(x,\nu,k)$-set if $X$ is a $(x,\nu,\sigma,k)$-set except that we do not require that $|X|< \sigma$. Similarly, 
$X$ is a $(x,\sigma,k)$-set if we do not require that $\vol(X)<\nu$.

\paragraph{The Framework.}
Our algorithm is based on the framework of  \cite{chechik2017faster}, so we start by explaining their ideas.
Their first crucial observation is that the recursive-mincut algorithm can take $\Omega(mn)$ time only when the
recursion depth is $\Omega(n)$.
This, in turn, happens when many cuts $S$ of value less than $k$ found are highly unbalanced, say $|S|=O(1)$. The main technical message of \cite{chechik2017faster} is that,
given a sub-linear time algorithm for finding such unbalanced cuts
$S$, one can improve the $O(mn)$ running time.

More precisely, their reduction shows that, given a vertex $x$, if we can check whether there exists
a $(x,\nu,k)$-set in $O(\nu)$ time, then the maximal $k$-edge-connected
subgraphs can be computed in $O(m\sqrt{m})$ time. In fact, their
reduction works even if we slightly relax the algorithm guarantee
to ``either report that $(x,\nu,k)$-sets do not exists or return
a $(x,O(\nu k),k)$-set''. Chechik~et~al.~\cite{chechik2017faster}
solves this relaxed problem in $O(\nu k^{O(k)})$ time and Forster~et~al.~\cite{forster2020computing}
improved it to $\Otil(\nu k^{2})$ time, which implies their final
algorithms with running time of $\Otil(m^{3/2}k^{O(k)})$ and $\Otil(m^{3/2}k^{3/2})$,
respectively.\footnote{Their improved running time in undirected graphs used techniques and
arguments that are specific to unweighted graphs. So these techniques
do not work in our setting.} Unfortunately, these algorithms are too slow when $k$ is large. 

We first observe that the reduction framework \cite{chechik2017faster}
can be adjusted so that it is compatible with our local algorithm
of \Cref{thm:local main} whose running time is independent from $k$.
More precisely, although the relaxed guarantee of ``either report
that $(x,\nu,k)$-sets do not exists or return a $(x,O(\nu k),k)$-set''
is not compatible with \Cref{prob:minextset}, we show that local algorithms
for \Cref{prob:minextset} can still be applied to break the $O(mn)$ barrier. The basic idea why algorithms for \Cref{prob:minextset} is sufficient follows from the fact that if a graph is not $k$-edge-connected, then either there must exist an $(x,\nu,k)$-extreme set (not just an $(x,\nu,k)$-set) for some vertex $x$, otherwise all cuts of size less than $k$ has volume at least $\nu$, i.e., it is quite balanced, which is a good case for us (see \Cref{lem:nubalanced} for the formal proof).

Since the framework of \cite{chechik2017faster} can be made compatible with local algorithms for \Cref{prob:minextset}, we now simply apply our local algorithms with  running time $\Otil(\nu\sigma^{2})=\Otil(\nu^{3})$ from \Cref{thm:local main}(\ref{enu:local vol}) or $\Otil(\sigma^{4})$
from \Cref{thm:local main}(\ref{enu:local size}) and  obtain our final
algorithm with running time $\Otil(mm^{3/4})$ or $\Otil(mn^{4/5})$,
respectively. This concludes \Cref{thm:main}.

\paragraph{Local Cut Algorithms.}
First, we recall the equivalent view of the random contraction technique based on minimum spanning trees (as observed in Karger's original paper \cite{karger1993global}). Assume for a moment that the input graph $G=(V, E)$ is unweighted. 

Define a random rank function $r:E\rightarrow[0,1]$ where the rank $r(e)$ of each edge $e$ is an independent uniform random number in $[0,1]$. Let $\MST_{r}$ be the minimum spanning tree w.r.t.~rank function $r$. 
\begin{lem}[\cite{karger1993global}]
\label{lem:key}For any minimum cut $(S,\overline{S})$ in $G$, with probability $\Omega(1/n^{2})$, 
\[
r(e')>r(e)
\]
for all $e'\in E(S,\overline{S})$ and $e\in\MST_{r}\setminus E(S,\overline{S})$.\footnote{For readers familiar with Karger's analysis, the ranks of edges correspond to the ordering of edge contraction. This condition is equivalent to saying that after contracting until there are only two vertices left, the unique remaining cut in the contracted graph is exactly the minimum cut $(S,\overline{S})$.}
\end{lem}

Next, we sketch how to prove \Cref{thm:local main}. Suppose that there is a $(x,\nu,\sigma,k)$-extreme set $S$ in $G$. Our goal is to return some $(x,\nu,\sigma,k)$-set. Note that $S$ is the unique minimum cut of the graph $G/\overline{S}$ by definition of an extreme set. By simply applying \Cref{lem:key} to the graph $G/\overline{S}$, which contains at most $\sigma$ vertices, we have the following:

\begin{lem}\label{lem:key local}
With probability $\Omega(1/\sigma^{2})$, 
\begin{equation}
r(e')>r(e)\label{eq:nice rank}
\end{equation}
for all $e'\in E(S,\overline{S})$ and $e\in G[S]\cap\MST_{r}\setminus E(S,\overline{S})$. 
\end{lem}

We say that the rank function $r$ is \emph{nice} with respect to a set $S$ or \emph{$S$ respects $r$} when this condition holds. We will generate $\Otil(\sigma^{2})$ many independent rank functions so that one of them is nice with high probability by \Cref{lem:key local}.
Now, suppose that $r$ is nice and suppose we run Prim's minimum spanning tree algorithm starting from the vertex $x$ to grow a ``partial'' MST $T_{r}\subseteq\MST_{r}$. We know from \Cref{eq:nice rank} that $T_{r}$ must span $S$ first before crossing the cut $(S,\overline{S})$. At that point $V(T_{r})=S$ is a $(x,\nu,\sigma,k)$-set and we can return $V(T_{r})$. 

Let $X=V(T_{r})$. Our algorithm is to simply return $X$ whenever $X$ is a $(x,\nu,\sigma,k)$-set. We can check this by maintaining $\vol(X)$, $|X|$ and $\delta(X)$. If $r$ is not nice, then $X$ might not ever be a $(x,\nu,\sigma,k)$-set but we will just terminate whenever $X$ is too big, i.e., $\vol(X)\ge\nu$ or $|X|\ge\sigma$.

It is simple to maintain $\vol(X)$, $|X|$ and $\delta(X)$ in $\Otil(\nu)$ time given the access to the adjacency lists of $G$. We can also maintain this information in $\Otil(\sigma^{2})$ time by using some simple pre-processing on $G$ (see \Cref{sec:localalg2} for details). Since there are $\Otil(\sigma^{2})$ many rank functions, the total running time becomes  $\Otil(\nu\sigma^{2})$ and $\Otil(\sigma^{4})$, respectively. This is how we obtain \Cref{thm:local main} when the graph is unweighted.

Lastly, we remove the assumption that the graph is unweighted. We observe that only the correctness of \Cref{lem:key} and \Cref{lem:key local} is based on the fact that the graph is unweighted. 
In weighted graphs, the rank function needed to be defined differently. Given a weighted graph $G = (V, E,w)$, if we treated $G$ as an unweighted multigraph and generated an independent random rank in $[0,1]$ for each edge, then \Cref{lem:key} would work, but this simulation is too slow. What we need is as follows: if an edge $e$ has weight $w(e)$, then we should define the rank $r(e)$ to be precisely the minimum of $w(e)$ independent random numbers in $[0,1]$. We derive this formula explicitly (see \Cref{lem:rankfn}) and show that it can be computed efficiently and also numerically stable (see \Cref{sec:invprobdist}), which is the only place we need that the edge weights are integral. This allows us to work in weighted graphs efficiently. This concludes the sketch of \Cref{thm:local main}.

\paragraph{Conditional Lower Bounds.}
For the first statement of \Cref{thm:lb main}, we show that the \emph{$t$-clique} problem can be reduced to the strengthened version of \Cref{prob:minextset}, i.e., the problem of checking whether there exists a $(x,\nu,\sigma,k)$-set for a given seed $x\in V$ even when $\sigma$ only depends on $t$. See \Cref{sec:8.3} for details. 
The similar reduction was shown in \cite{fomin2013parameterized} but in their reduction $\sigma = \Omega(n)$ is quite big, since they work in unweighted graphs.

Next, we explain the proof of the second statement of \Cref{thm:lb main} based on the OMv conjecture \cite{henzinger2015unifying}. In \cite{henzinger2015unifying} they show that, assuming the conjecture, given a graph $G = (V,E)$ with $n$ vertices and $\poly(n)$ preprocessing time on $G$, then given a query $S \subset V$, there is no algorithm that checks whether $S$ is an independent set in time $n^{2-\Omega(1)}$. 
We will show how to use a subroutine for \Cref{prob:minextset} to check if $S$ is independent. Given $S$, imagine we artificially add a cycle $C_S$ spanning $S$ where every edge $e \in C_S$ has infinite weight. Let $G_S$ be the resulting graph. 
Let $\sigma = |S|+1$ and $k = \vol(S)$. Fix an arbitrary vertex $x \in S$.
If $S$ is not independent, then $S$ is an $(x,\sigma,k)$-extreme set in $G_S$ because $|S|<\sigma$,  $\delta(S)<\vol(S) = k$, and any strict subset $S' \subset S$ has infinite cut value.
On the other hand, if $S$ is independent, then there is no $(x,\sigma,k)$-set in $G_S$. Suppose for contradiction that there is a $(x,\sigma,k)$-set $T$. We have that $T$ cannot cross $S$, otherwise, $T$ has an infinite cut value. Since $x \in T$, so $S\subseteq T$. But $|T| \le \sigma -1 = |S| $, so $T = S$. Now, we have  $\delta(T)=\delta(S)=\vol(S)=k$ because $S$ is independent, which is contradiction.
Lastly, note that we can easily simulate the adjacency list and matrix of $G_S$ given those of $G$ with just an additive $O(\sigma)$ overhead for any query set $S$ which gets subsumed in $O(\sigma^{2-\eps})$ time.
Therefore, given the subroutine for \Cref{prob:minextset} with access to adjacency list and matrix of $G$ with running time $\sigma^{2-\Omega(1)}$, we can then check if $S$ is independent in $n^{2-\Omega(1)}$, refuting the OMv conjecture.

\section{Preliminaries}\label{sec:prelims}

We refer to the paragraph above the statement of \Cref{prob:minextset}
for basic notations. Given a weighted undirected graph $G=(V,E,w)$
where $w\in\mathbb{Z}_{\ge0}^{E}$, $G$ is $k$-edge-connected if global minimum cut value,
$\min_{\emptyset\neq S\subset V}w(\delta(S))\ge k$. We give the proof
of the following simple observation in the \Cref{sec:mcpunique}. 
\begin{prop}
\label{prop:mcpunique}The maximal $k$-edge-connected subgraphs $\{V_{1},\dots,V_{z}\}$
of any graph $G$ are unique and form a partition of $V$.
\end{prop}

By \Cref{prop:mcpunique}, we will usually call our desired solution $\{V_{1},\dots,V_{z}\}$
as \emph{maximal $k$-edge-connected partition} to emphasize that it forms a partition.

\subsection{The Recursive-Mincut algorithm}\label{sec:naivealgo}
In this subsection, we present a naive algorithm (see e.g.~the appendix of \cite{chen2020near}) to find the maximal $k$-edge-connected partition of the graph as it helps the reader to familiarize with the ideas involved in our improved algorithm and serves as a warmup for the proofs that follow.

Suppose we have a cut in the graph $G$ that is less than $k$, then any maximal vertex set $k$-edge-connected should be contained on one side of the cut. Otherwise, we can separate the graph induced on that vertex set with edges of weight at most $k$; thus, it is not $k$-edge-connected and thus not maximal $k$-edge-connected.

Hence, a straightforward algorithm finds cuts of size less than $k$ if they exist and recur on both the graphs induced by the cut. If no cut of size less than $k$ exists, then the graph is $k$-edge-connected, and we add the remaining vertices to the partition, and we are done.

\begin{algorithm}[H]
\KwData{$G = (V,E,w),k$}
\KwResult{Maximal $k$-edge-connected partition of vertex set $V$}
Compute minimum cut $(S,V\setminus S)$ of graph $G$\;
\eIf{$w(S,V\setminus S) \geq k$}
{return $\{V\}$}
{return $\nmp(G[S],k) \cup \nmp(G[V\setminus S],k)$}
\caption{\nmp(G,k)}
\label{alg:naive}
\end{algorithm}

\begin{lemma}\label{lem:naive:correctness}
The output of \Cref{alg:naive}
is the maximal $k$-edge-connected partition of $G$.
\end{lemma}

\begin{proof}
Let $\cP$ be the output of the algorithm. Observe that the sets output by \Cref{alg:naive} form a partition of the vertex set $V$ as no vertex is repeated and every vertex is included as the recursion would eventually bottom out at least with singleton vertex, because each vertex is $k$-edge-connected trivially.

Now, we argue that the subgraph induced on each set $S \in \cP$ is a maximal $k$-edge-connected subgraph. 
Indeed, $G[S]$ $k$-edge-connected by how we stop the recursion. Next, we show the maximality of $S$.
Assume for contradiction a set $T\supset S$ such that $G[T]$ is also $k$-edge-connected. At the start of the algorithm, set $T$ is contained in $V$. Let $V'$ be the last set during the algorithm's run on $G$ that contains $T$. So the minimum cut of $G[V']$, of size less than $k$, separates the set $T$ into two parts. Hence, $G[T]$ also has a cut of size less than $k$, which contradicts our assumption that $G[T]$ is $k$-edge-connected. Therefore, we conclude that $G[S]$ is a maximal $k$-edge-connected subgraph.
\end{proof}

\begin{lemma}
\Cref{alg:naive} takes at most $\Tilde{O}(mn)$ time.
\end{lemma}

\begin{proof}
The recursion depth of \Cref{alg:naive} is clearly at most $n$ because each recursion reduces the number of vertices in the largest subgraph by at least one.

We use the Karger's algorithm from \cite{karger2000minimum} that takes a near-linear running time $O(m\log^3 n)$ to compute a minimum cut in the graph with $m$ edges and $n$ vertices. Since all the subgraphs at a particular recursion depth are disjoint and have at most $m$ edges in total. Let $G_i$ be one of the subgraph with $m_i$ edges and $n_i$ vertices.  Hence, the total time to compute minimum cut on all the subgraphs is $\sum_i O(m_i\log^3 n_i) \leq \sum_i O(m_i\log^3 n) \leq \Otil(m)$. So, the total running time over all the recursion levels is at most $\Otil(mn)$.
\end{proof}

\section{Local algorithm for small volume cuts}\label{sec:extsets}

The goal of this section is to prove \Cref{thm:local main} (\ref{enu:local vol}). We organize this section as follows. We introduce the key subroutine, \lp of our local algorithm in \Cref{sec:4.1}. Given a \emph{nice} rank function $r$ with respect to an extreme set $S$ (recall \Cref{lem:key local}), \lp would find an $(x,\nu,k)$-set. In \Cref{sec:4.2} we introduce the idea of random rank function and prove \Cref{lem:key local} using the idea of random contraction from \cite{karger1993global}. From \Cref{lem:key local} we know that a random rank function is nice with probability $\Omega(1/\sigma^2)$. In \Cref{sec:4.3}, we give our local algorithm which repeats the \lp sub-routine with multiple random rank functions to amplify the probability of finding a \emph{nice} rank function, which then solves the \Cref{prob:minextset} with high probability and prove \Cref{thm:local main} (\ref{enu:local vol}). Here, we restate the first part of the \Cref{thm:local main}.

\begin{theorem}\label{thm:sec4}
There exists a randomized algorithm that solves \Cref{prob:minextset} with high probability in time $O(\nu\sigma^2\log^2 n)$ given the adjacency lists of the graph.
\end{theorem}

\subsection{Local Prim Algorithm}\label{sec:4.1}
In this sub-section, we introduce the key sub-routine required for our local algorithm.
To this end, we need to introduce some definitions. A \emph{rank function} $r:E\rightarrow [0,1]$ is a function that maps each edge to a number in $[0,1]$. We call $r(e)$ the \emph{rank of edge $e$}. Let $\MST_r(G)$ denote the minimum spanning tree of $G$ with respect to the rank function $r$.

As we are always interested in the value of the rank $r(e)$ in comparision to other edges rather than the absolute rank of an edge, in the rest of the paper we call the rank value as the rank in short.

\begin{definition}[Set respects rank function]
Given a graph $G=(V,E,w)$ and a rank function $r:E\rightarrow [0,1]$, a set $S \subset V$ \emph{respects} the rank function $r$ if 
$\MST_r(G) \cap E(S,\overline{S}) = {e'}$ and $r(e') > r(e)$ for all $e \in \MST_r(G) \cap G[S]$.
\end{definition}

In words, the set $S$ respects $r$ if the $\MST_r$ has only one edge $e'$ in intersection with the cut $\delta(S)$, which makes $\MST_r \cap G[S]$ a minimum spanning tree of $G[S]$, and each edge of such minimum spanning tree contained in $G[S]$ should have rank lesser than the rank of edge $e'$. Recall the discussion below \Cref{lem:key local}.
We say a rank function $r$ is \emph{nice} w.r.t.~a set $S$ iff $S$ respects $r$.

This definition is very critical to the correctness of our algorithm. Note that the minimum spanning tree used for analysis is with respect to rank function $r$. It can be very different from the one computed using the weight function $w$.

We consider the following problem.

\begin{restatable}[\localmst]{problem}{localmstprob}
\label{prob:localmst}
Given a graph $G=(V,E,w)$, a vertex $x$, a rank function $r$ and parameters $\nu,\sigma,k$ find a $(x,\nu,\sigma,k)$-set $S$ that respects the rank function $r$ or return $\bot$ if no such $(x,\nu,\sigma,k)$-set respecting $r$ exists.
\end{restatable}

We state the Local Prim algorithm that solves the \Cref{prob:localmst} and proves its correctness.

\begin{algorithm}[H]
\KwData{$G = (V,E,w),x,\nu,\sigma,k,r$ where $r$ is a rank function $r:E\rightarrow [0,1]$}
\KwResult{$(x,\nu,\sigma,k)$-set respecting $r$ if exists else $\bot$}
$X = \{x\}$\;
\While{$w(\delta(X)) \geq k$ and $\vol(X) < \nu$ and $|X| < \sigma$}
{
    Find the edge with minimum rank, $e = (u,v) \in E(X,V\setminus X) = \delta(X)$\;\label{lp:line3}
    Update $X = X \cup \{v\}$\;
}
\eIf{$w(\delta(X)) < k$}
{\Return{$X$}}
{\Return{$\bot$}}
\caption{\lp}
\label{alg:localprim}
\end{algorithm}

We expand the set $X$ starting from the singleton vertex cut $\{x\}$. Since we only want the sets that respect the rank function, we consider the cut edge $(u,v)$ with minimum rank where $u\in X$ and $v \notin X$, and expand $X$ by setting $X \gets X \cup  \{v\}$.
We expand the set $X$ along the edges with minimum rank until we violate volume or cardinality constraints of the set $X$ or find a set with a cut size less than $k$. If there exists a $(x,\nu,\sigma,k)$-set $S$ that respects the rank function $r$, then \Cref{alg:localprim} finds the minimal set such that no strict subset of it also respects the rank function $r$. We prove it in the following lemma.

\begin{lemma}\label{lem:lpreturnsS}
If there exists a $(x,\nu,\sigma,k)$-set $S$ in $G=(V,E,w)$ that respects a rank function $r$ and is minimal such that no strict subset respects $r$, then \Cref{alg:localprim} returns $S$ and return $\bot$ only if no such respecting set $S$ exists.
\end{lemma}

\begin{proof}
Let $S'$ be the set returned by the \Cref{alg:localprim}. We claim $S'$ cannot contain a vertex not in $S$. Since we start from $X =\{x\}\subset S$ and so, to get a vertex not in $S$, we need to expand set $X$ along an edge from $\delta(S)$.

Let $X$ be the current set before expanding along an edge from $\delta(S)$. If $X \subsetneq S$ then there exists an edge in $\delta(X)$ that is part of the minimum spanning tree of $G[S]$ that has smaller rank than any edge in $\delta(S)$ as $S$ respects $r$. So $X = S$, before we can contract an edge from $\delta(S)$, but if $X=S$, $w(\delta(X))=w(\delta(S))<k$ as $S$ is a $(x,\nu,\sigma,k)$-set so \lp terminates without expanding along an edge from $\delta(S)$ and returns $S$. So $S'$ cannot contain a vertex outside $S$ and thus $S'\subseteq S$. Since any subset of $S$ is also $(\nu,\sigma)$-set and according to the algorithm $X$ always contains $x$. So we only terminate when cut size is less than $k$ and thus any set returned is a $(x,\nu,\sigma,k)$-set.

It remains to prove that \lp{} cannot return $\bot$. Assume it returns $\bot$. We only return $\bot$ when we cannot find any $(x,\nu,\sigma,k)$-set. However, since $S$ respects $r$, expanding along minimum ranked edges is nothing but applying Prim's algorithm starting from $x$. According to the set respecting rank function definition, we expand along all the edges of the minimum spanning tree in $G[S]$ without expanding the set $X$ along any edge from $\delta(S)$ as they all have a higher rank, making $X=S$. Thus \lp{} either returns $S$ as it exits the while loop or terminates early if it finds a subset that also respects $r$ but cannot return $\bot$.
\end{proof}

We now bound the running time of the \lp{} algorithm.

\begin{lemma}\label{lem:lpruntime}
\lp{} subroutine takes at most $O(\nu \log \nu)$ time.
\end{lemma}

\begin{proof}
Here, we describe how to implement \lp. 
We create a dictionary $D$ with an invariant that $D = E(X,V\setminus X)$ contains all the cut edges crossing $X$.
Given this dictionary, we can obtain a minimum rank edge $e = (u,v) \in E(X,V\setminus X)$ for \Cref{lp:line3} of \Cref{alg:localprim}. This look-up operation occurs at most $|X| < \sigma$ times as we terminate whenever $|X| \ge \sigma$.

Whenever we expand $X \gets X \cup \{ v \}$, we update the dictionary by scanning through all edges of $v$ in $O(1+\deg(v))$ time. We delete from $D$ all edges between $v$ and $X$ and insert into $D$ all edges between $v$ and $V\setminus X$. We also update the cut value of $X$ while iterating over edges of $v$ by subtracting the weight of the edges that are already incident to $X$ and adding the weight of the edges that are not incident to $X$ previously.

The total number of insertion and deletion operations is $O(\sum_{ v \in X } (1+\deg(v))) = O(\sigma+\nu)$ since we terminate whenever $\vol(X) \ge \nu$ or $|X|>\sigma$. Since each operation of the dictionary takes logarithmic time, the total time is at most $O((\sigma+\nu)\log \nu) = O(\nu \log \nu)$.

\end{proof}

According to above \Cref{lem:lpreturnsS} given there is a $(\nu,\sigma,k)$-set $S$ in graph $G$. If we find a rank function $r$ such that $S$ respects $r$, then we can find a $(\nu,\sigma,k)$-set by running \lp{} with input as any vertex $x\in S$ using this rank function $r$ in $O(\nu\log \nu)$ time.

The following sub-section will show how we construct the rank functions that help us find $(\nu,\sigma,k)$-sets.

\subsection{\emph{Nice} Rank Function}\label{sec:4.2}
In this subsection, we show how to construct nice rank functions introduced above and explain how it is related to the idea of using the randomized contraction technique of \cite{karger1993global} to find minimum cuts.

Constructing a rank function respected by any particular set $S$ is easy. However, we do not know which subset is a $(\nu,\sigma,k)$-extreme set, so we cannot explicitly construct a rank function that is respected by an unknown extreme set. We come up with a distribution $R$ over rank functions and prove that any rank function chosen from this distribution is \emph{nice} with good probability. The key idea used in coming up with the distribution $R$ is the idea of random contraction.

To understand the construction of $R$ we need to look at the process of Karger's random contraction from a different perspective. We begin with the process of how a random edge is sampled in Karger's mincut algorithm. It can be equivalently viewed as stated in the Lemma 6.2.1 of \cite{karger1996new}.

\begin{claim}[Lemma 6.2.1 of \cite{karger1996new}]\label{lem:rankfn}
If we choose a rank function that maps each edge $e$ to $r(e) = 1-(1-t)^{1/w(e)}$ computed by uniformly sampling a value $t \in [0,1]$ where $w(e)$ is the weight of the edge, then the probability of an edge $e$ having the minimum rank among all edges equals $\frac{w(e)}{\sum_{e\in E} w(e)}$.
\end{claim}

For completeness, we give the proof of \Cref{lem:rankfn} with respect to our formalism and defer the proof to \Cref{sec:invprobdist}.
The rest of this section shows how a random rank function is a nice rank function with good probability.

From \Cref{lem:rankfn}, we know that choosing an edge that has the minimum rank value ($r(e)$) among all the edges simulates the process of picking an edge proportional to its weight. So Karger's random contraction \cite{karger1993global} algorithm is nothing but picking an edge with minimum rank $r(e)$ and contracting it. This process is the same as applying Kruskal's algorithm on the graph $G = (V, E, w)$ with $r$ as the weight function. So, the set of all edges contracted during the running of Karger's random contraction process form $\MST_r$ of the graph $G$.

In the following lemma, we prove that any $(x,\nu,\sigma,k)$-extreme set respects a random rank function described above with reasonable probability.

\begin{remark}
Any $(x,\nu,\sigma,k)$-set does not respect a random rank function with reasonable probability, and we need more structure for a set to respect a random rank function.
\end{remark}

We define formally the distribution of the rank functions we sample from. For each edge $e\in E$ of the graph $G=(V,E,w)$ we define a random variable $R_e$ with cumulative distribution function (CDF) as $\Pr(R_e \leq t) = 1- (1-t)^{w(e)}$. Distribution $R$ is the joint probability distribution of the random variables for all edges in $E$, $R = (R_{e_1},R_{e_2},\cdots)$. Note that each $R_{e}$ is independent of another edge's random variable $R_f$ when $e\neq f$.
So, choosing a random rank function $r$ from the distribution $R$ is nothing but each edge $e\in E$ choosing a value according to $R_e$ which is the rank value $r(e)$ for that edge.

We now prove a key lemma that is crucial for proving the correctness of our local algorithm.

\begin{lemma}\label{lem:rrrprob}
Let $r$ be a rank function sampled from the distribution $R$ and $S$ be any $(x,\nu,\sigma,k)$-extreme set in the graph $G = (V,E,w)$ then the probability that $S$ respects $r$ is at least $\frac{2}{\sigma(\sigma-1)}$.
\end{lemma}

\begin{proof}
We need to lower bound the probability that an extreme set $S$ respects the rank function $r$.

Let $G/\overline{S} = (V_S,E_S,w_S)$ be the graph which is obtained by contracting the set $V \setminus S$ into a single super node $t$. $V_S=S\cup \{t\}$. From definition of the $(x,\nu,\sigma,k)$-extreme set $(S,\{t\})$ is the unique minimum cut of $G/\overline{S}$.

Using Karger's random contraction algorithm on graph $G/\overline{S}$, the probability of successful contractions to find the minimum cut is at least $\frac{2}{\sigma(\sigma-1)}$ as number of vertices in the graph $G/\overline{S}$ are bounded by $\sigma$.

As noted in \Cref{lem:rankfn} sampling of an edge in Karger's random contraction process can be simulated using a rank function chosen from $R$ restricted on the graph $G/\overline{S}$ i.e. considering only the edges in $G/\overline{S}$. So rank function $r$ succeeds in finding the unique minimum cut $(S,\{t\})$ with probability at least $\frac{2}{\sigma(\sigma-1)}$ i.e. all the vertices of the set $S$ are contracted into a single node say $s$ and finally we have the cut $(\{s\},\{t\})$ as it is the unique minimum cut.

Now we show that such a succeeding rank function chosen from $R$ is respected by the set $S$ in graph $G/\overline{S}$. Since the rank function, $r$ should contract the whole $S$ without contracting any edge from the minimum cut $\delta_{G/\overline{S}}(S)$. The MST of $G/\overline{S}$ with respect to the rank function $r$ should cross the minimum cut $S$ in $G/\overline{S}$ only once and also should span $G[S]$ before crossing the set $S$.

So we conclude that all edges in the minimum cut $(S,\{t\})$ of $G/\overline{S}$ have a higher rank than the edges that belong to the spanning tree restricted to $G[S]$. So $S$ respects $r$ in $G$ whenever the random rank function finds the unique minimum cut in $G/\overline{S}$.

Since probability of success of random contraction process finding minimum cut in $G/\overline{S}$ is at least $\frac{2}{\sigma(\sigma-1)}$, probability of $S$ respecting random rank function is also at least $\frac{2}{\sigma(\sigma-1)}$.
\end{proof}

\subsection{Local Subroutine for \texorpdfstring{$(\nu,\sigma,k)$}{(nu,sigma,k)}-set}\label{sec:4.3}

In the \Cref{sec:4.1}, we have seen how to find a $(x,\nu,\sigma,k)$-set that respects a given rank function $r$ by using \lp{} sub-routine. \Cref{sec:4.2} showed that any $(\nu,\sigma,k)$-extreme set respects a random rank function with good probability. In this sub-section we combine both to give an algorithm that solves \Cref{prob:minextset} and proves \Cref{thm:local main}(\ref{enu:local vol}).

From \Cref{lem:rrrprob} we know that a random rank function is respected by an $(x,\nu,\sigma,k)$-extreme set with $\Omega(1/\sigma^2)$ probability. To find a rank function that is respected with high probability, we need to independently repeat this process $O(\sigma^2 \log n)$ times.

So the algorithm to solve \Cref{prob:minextset} would mainly be iterating over $(O(\sigma^2 \log n))$ many independent random rank functions. Since the $(x,\nu,\sigma,k)$-extreme set respects at least one random rank function among $O(\sigma^2 \log n)$ and so we are guaranteed to return a $(x,\nu,\sigma,k)$-set by \Cref{lem:lpreturnsS}.

We formalize the above argument and give an algorithm that solves \Cref{prob:minextset}. We can fix all the random rank functions $r^i, i \in [O(\sigma^2\log n)]$ before the algorithm starts with out loss of generality.

\begin{algorithm}[H]
\KwData{$G = (V,E,w),x,\nu,\sigma,k$} 
\KwResult{$(x,\nu,\sigma,k)$-set or $\bot$ if no $(x,\nu,\sigma,k)$-extreme set exists}
\For{$ i \in [O(\sigma^2 \log n)]$}
{
    $X = \lp(G,x,\nu,\sigma,k,r^i)$\;\label{alg:blp:3}
    \If{$X \neq \bot$}
    {
        \Return $X$\;
    }
}
\Return $\bot$\;
\caption{\blp}
\label{alg:blp}
\end{algorithm}

\begin{lemma}\label{lem:blp:correctness}
\Cref{alg:blp} solves \Cref{prob:minextset} correctly with high probability.
\end{lemma}

\begin{proof}
From \Cref{lem:rrrprob} we know that any $(x,\nu,\sigma,k)$-extreme set should respect a random rank function with probability at least $\frac{2}{\sigma(\sigma-1)}$. If there exists an $(x,\nu,\sigma,k)$-extreme set $S$, then $S$ will respect at least one of the $\Theta(\sigma^2 \log n)$ random rank function. The probability that none of them is respected by $S$,

\[
\Pr(S\text{ does not respect any } r^i) \leq \Big(1-\frac{2}{\sigma(\sigma-1)}\Big)^{\Theta(\sigma^2 \log n)}\\
\leq \frac{1}{n^{\Theta(1)}}\\
\]

Since there is a rank function that is respected by the $(x,\nu,\sigma,k)$-extreme set. According to \Cref{lem:lpreturnsS} we find a $(x,\nu,\sigma,k)$-set and we are done.

However if we do not find a $(x,\nu,\sigma,k)$-set by any of the random rank functions then we know that no $(x,\nu,\sigma,k)$-extreme set existed and we are only wrong with negligible probability $\frac{1}{n^{\Theta(1)}}$.
\end{proof}

\begin{lemma}\label{lem:blp:runtime}
Rutime of the \Cref{alg:blp} is
\begin{equation}
T(\blp) = O(\sigma^2\log n) \cdot T(\lp).\label{eqn:tblp}
\end{equation}
\end{lemma}

\begin{proof}
Each iteration of the for loop takes $T(\lp)$ time and so the total time taken by \blp{} is $O(\sigma^2 \log n)\cdot T(\lp)$.
\end{proof}

We now prove the \Cref{thm:sec4} which is same as proving \Cref{thm:local main}(\ref{enu:local vol}).

\begin{proof}[Proof of \Cref{thm:sec4}:]
From \Cref{lem:lpruntime} we have $T(\lp) = O(\nu\log \nu)$ which is at most $O(\nu\log n)$. According to \Cref{lem:blp:runtime}, $T(\blp) = O(\nu\sigma^2\log^2 n)$. From \Cref{lem:blp:correctness} we have the \Cref{alg:blp} that solves \Cref{prob:minextset} with high probability in running time at most $O(\nu\sigma^2\log^2 n)$ proving \Cref{thm:sec4}.
\end{proof}

Since any $(\nu,\sigma,k)$-extreme set of volume at most $\nu$ cannot have more than $\nu$ vertices as they need to be connected hence a trivial upper bound on the number of vertices contained in the set would be at most $\nu$. Substituting the trivial upper bound for number of vertices $\sigma$ in \Cref{eqn:tblp} we get $T(\blp) = O(\nu^3 \log^2 n)$. We can now eliminate the parameter $\sigma$ from \Cref{thm:sec4}, which gives the following theorem.

\begin{theorem}\label{thm:blp}
There exists a randomized algorithm given a graph $G=(V,E,w)$, a vertex $x$ and parameters $\nu,k$ finds either a $(x,\nu,k)$-set or guarantees that no $(x,\nu,k)$-extreme set exists with high probability in time $O(\nu^3\log^2 n)$.
\end{theorem}

\section{\texorpdfstring{$\Tilde{O}(m^{1.75})$}{m1.75} algorithm for maximal \texorpdfstring{$k$}{k}-connected partitions}\label{sec:recdepth}

In the previous section, we introduced the notion of small local cuts. We also gave a local algorithm that finds small volume cuts whose cut size is smaller than a given parameter $k$. This section introduces an algorithm that uses the local sub-routine to improve recursion depth and finds maximal $k$-edge-connected partition in running time $\Otil(m^{1.75})$. Although this algorithm can be slower than $\Otil(mn)$ when the graph is dense, in the next section, we introduce another algorithm that improves the running time $\Otil(mn)$ for finding maximal $k$-edge-connected partitions.

The main goal of this section is to prove the first running time stated in \Cref{thm:main}. We restate it here for completeness.

\begin{theorem}\label{mcp:mainthm}
There exists a randomized algorithm that takes a graph $G=(V,E,w)$ with $m$ edges and $n$ vertices and finds the maximal $k$-edge-connected partition with high probability in at most $O(m^{1.75}\log^{3.75} n)$ time.
\end{theorem}

\subsection{Global Algorithm using \blp}
Let us see how to use the \blp{} subroutine described above to improve the recursion depth. For the rest of this section, we omit the cardinality constraint $\sigma$. The main idea is to find as many $(\nu,k)$-sets as possible using \blp. However, it might happen that removing a $(\nu,k)$-set might give rise to new extreme sets in the graph. So, at every stage, we maintain a list of candidate vertices ($L$ as given in \mcp) which might be inside a $(\nu,k)$-extreme set. We remove all these extreme sets by using \blp{} with input as vertices $x$ from the candidate list, which either finds a $(\nu,k)$-set containing $x$ or guarantees that no extreme set containing $x$ exists.

Exhausting the candidate list ensures the residual graph has no small extreme sets, which would imply a $\nu$-balanced minimum cut. Let us prove this below.

\begin{lemma}\label{lem:nubalanced}
If there are no $(\nu,k)$-extreme sets in the graph $G=(V,E,w)$, then either
\begin{itemize}
    \item minimum cut is at least $k$ in which case graph $G$ is $k$-edge-connected.
    \item minimum cut is less than $k$ in which case it is $\nu$-balanced.
\end{itemize}
\end{lemma}

\begin{proof}
The first part is trivial; if the graph's minimum cut is at least $k$, it is $k$-edge-connected.

If the minimum cut of the graph is less than $k$ and if every inclusion-wise minimal cut is $\nu$-balanced then we are done. Assume there exists an inclusion-wise minimum cut $(A,\overline{A})$ that is not $\nu$-balanced. Without loss of generality $vol(A)<\nu$. Because $A$ is inclusion wise minimum cut, $A$ is the unique minimum cut in $G/\overline{A}$. So $A$ is a $(\nu,k)$-extreme set which is contradiction. So $vol(A)\geq \nu$. So every inclusion wise minimal cut has volume $\geq \nu$ and so every minimum cut is $\nu$-balanced.
\end{proof}

From above \Cref{lem:nubalanced} we have proved that if there are no $(\nu,k)$-extreme sets in the graph and the graph is not $k$-edge-connected, then every minimum cut is $\nu$-balanced and thus remove at least $\nu$ volume from the graph. It only takes at most $O(m/\nu)$ recursion depth for the maximum sized sub-graph to reduce to half the original volume $m/2$. Given the high-level overview of the algorithm, we now formally define and analyze it below.

\begin{algorithm}[H]
\KwData{$G = (V,E,w),\nu,\sigma,k,L,m$}
\KwResult{$k$-cut partition of $G$}
$\hat{G} = G$\;
$\mathcal{R} = \{\}$\;
\While{There exists $x \in L$}
{
    $S=\blp(\hat{G},x,\nu,\sigma,k)$\label{mcp:3}\;
    \eIf{$S\neq \bot$}
    {
    $\mathcal{R} = \mathcal{R} \cup \{S\}$\;
    $L=(L \cup V(\delta_{\hat{G}}(S)) \setminus S$\;
    $\hat{G} = \hat{G} \setminus S$\;
    }
    {
    $L = L\setminus \{x\}$\;
    }
}
$(A,B) = \mcut(\hat{G})$\label{mcp:10}\;
\uIf{$\lambda(G)\geq k$}
{$\mathcal{R} = \mathcal{R}\cup \{\hat{V}\}$\;}
\uElseIf{$vol(\hat{G}[A]) > m$}
{$\mathcal{R} = \mathcal{R} \cup \kcp(\hat{G}[A],\nu,\sigma,k,V(\delta_{\hat{G}}(A))\cap A,m)$}\label{alg:kcp:15}
\uElseIf{$vol(\hat{G}[B]) > m$}
{$\mathcal{R} = \mathcal{R} \cup \kcp(\hat{G}[B],\nu,\sigma,k,V(\delta_{\hat{G}}(B))\cap B,m)$\;}\label{alg:kcp:17}
\Else{$\mathcal{R} = \mathcal{R} \cup \{A,B\}$\;}
\Return $\mathcal{R}$\;
\caption{\kcp}
\label{alg:kcp}
\end{algorithm}

\kcp{} partitions the graph $G = (V,E,w)$ using cuts of size less than $k$, but doesn't give us the maximal $k$-edge-connected partition. It reduces the problem to smaller-sized sub-graphs. This is proved in the following \Cref{lem:kcp:correctness}.

\begin{lemma}\label{lem:kcp:correctness}
Let $\mathcal{P}$ be the maximal $k$-edge-connected partition of $G$ and $\mathcal{R}$ be the result returned by \Cref{alg:kcp}. For all $U \in \mathcal{P}$, $\exists U' \in \mathcal{R}$ such that $U \subset U'$. In other words, $\mathcal{P}$ is a refinement of the partition $\mathcal{R}$. For all $U' \in \mathcal{R}, \vol(G[U']) = \vol(G)/2$.
\end{lemma}

\begin{proof}
As seen in the proof of \Cref{lem:naive:correctness}, any maximal $k$-edge-connected set of vertices $U$ cannot be separated by a cut of size less than $k$. Every cut used to partition the graph in \kcp{} is of size less than $k$. Hence every maximal $k$-edge-connected set of $G$ has to be strictly contained in one set belonging to $\mathcal{R}$.

Since we always recurse on the larger volume side of the minimum cut, the smaller side has less than half the original volume. Any subgraphs formed by local cuts have volume less than $\nu$, which we later choose to be less than $m/2$. Thus all subgraphs have at most $\vol(G)/2$.
\end{proof}

Note that every maximal $k$-edge-connected set of the original graph is also a maximal $k$-edge-connected set of the smaller sub-graph resulting from \kcp{}. Thus naturally, the algorithm for finding the maximal $k$-edge-connected partition is as follows.

\begin{algorithm}[H]
\KwData{$G = (V,E,w),k$}
\KwResult{Maximal $k$-edge-connected partition of $G$}
$\mathcal{R} = \kcp(G,\nu(|E|),\sigma(|V|),k,V,|E|)$\;
$\mathcal{P} = \{\}$\;
\For{$U \in \mathcal{R}$}
{$\mathcal{P} = \mathcal{P} \cup \mcp(G[U],k)$\;}
\Return $\mathcal{P}$\;
\caption{\mcp}
\label{alg:mcp}
\end{algorithm}

The size of the graphs for which we find the maximal $k$-edge-connected partition is at most half the original volume. So, after at most $O(\log m)$ recursive applications of \kcp{} we find every maximal $k$-edge-connected set $U \in \mathcal{P}$. Since any graph that is $k$-edge-connected cannot broken further by \kcp{}, or we end up with a single vertex which is trivially $k$-edge-connected. In the following \Cref{lem:mcp:correctness} we prove the correctness of the \Cref{alg:mcp}.

\begin{lemma}\label{lem:mcp:correctness}
\Cref{alg:mcp} computes the maximal $k$-edge-connected partition of the graph $G$.
\end{lemma}

\begin{proof}
From \Cref{lem:kcp:correctness}, we have seen that any maximal $k$-edge-connected set is strictly contained in one set of the partition returned by \kcp{}. However, according to the guarantee of \kcp{} volume of each component is reduced by half of the original graph. Hence every component has volume at most $m/2^i$ at a recursion depth of $i$. So any particular set $S$ from maximal $k$-edge-connected partition with $vol(G[S]) \in [\frac{m}{2^{i+1}},\frac{m}{2^i}]$ is identified atleast at a recursion depth of $i+1$, as none of the cuts break the set $S$ and so after $O(\log n)$ depth each component either has become $k$-edge-connected and so has become part of the partition $\mathcal{P}$ or the volume becomes $0$ i.e., a singleton vertex which is trivially $k$-edge-connected.
\end{proof}

\subsection{Runtime analysis}
In this subsection, we analyze the run time of the \Cref{alg:kcp} and then \Cref{alg:mcp}.

We divide the running time of \kcp{} into two parts: the invocation of the local subroutines and the invocations of the global minimum cut algorithm. For the first part, since we only apply local subroutine with seed vertices from the candidate list $L$ in graph \kcp{}. We can bound the running time by bounding the number of vertices added to $L$. We start with candidate list $L = V$ initially as $(\nu,k)$-extreme sets can contain any vertex in $V$.

Once a cut of size less than $k$ is found (using either \blp{} or \mcut{}), we remove the smaller side from the graph $\hat{G}$ and find cuts in the residual graph. Because we are removing some edges from the graph $\hat{G}$, it can lead to the origin of new extreme sets that might not be present before. So we need to update $L$ to account for new extreme sets that might arise due to the removal of the cut edges. It is unwarranted and time-consuming to iterate over all vertices to find new extreme sets. Similar to the idea from \cite{chechik2017faster}, we prove that any new extreme sets that arise due to the removal of an edge from the graph should contain at least one of the edge's endpoints. We add this vertex to the candidate list so that we can find a $(\nu,k)$-set surrounding this vertex if the new extreme set is $(\nu, k)$-extreme set.

\begin{claim}\label{lem:newcandidates}
Any new extreme sets introduced in the graph by removing an edge must contain at least one endpoint of that edge.
\end{claim}

\begin{proof}
Let $G$ be the graph before removing the edge $e$ and $H$ be the graph after removing. Removal of an edge $e$ may lead to the formation of new extreme set $A$ only if it belongs to the cut edges of the new extreme set i.e., $e\in \delta_G(A) = \delta_{G/\overline{A}}(A)$. If it is an internal edge of the new extreme set, i.e., $e \in G[A]$, then even before the removal of the edge, it is an extreme set, and hence it is not formed by removing the edge $e$. If it belongs to $G[\overline{A}]$, it does not affect whether $A$ is an extreme set or not as we consider the graph $G/\overline{A}$ in deciding whether $A$ is an extreme set. Since it is the cut edge of the extreme set, then there exists an endpoint of the edge $e$ in the set, proving the claim.
\end{proof}

Once the candidate list is exhausted, we are sure (with high probability) that there are no extreme sets in the residual graph; thus, we apply the global minimum cut to get a $\nu$-balanced cut and recurse on the larger side, as shown in \Cref{alg:kcp:15,alg:kcp:17}. We now bound the recursion depth of \kcp{} in the below lemma.

\begin{claim}\label{alg:kcp:recdepth}
\Cref{alg:kcp} has a recursion depth of at most $O(m/\nu)$ for a $m$ edge graph.
\end{claim}

\begin{proof}
At the end of the while loop, we have the guarantee that there are no $(\nu,k)$-extreme sets, which means from \Cref{lem:nubalanced} we either have that the minimum cut is at least $k$ in which case we are done by adding the vertex set of the current graph to the $\mathcal{R}$. Otherwise, the minimum cut is $\nu$-balanced. So the largest component is decreased by at least $\nu$ volume when a global mincut is called. Since initially we have volume $m$, and at each recursion depth, we reduce the volume of the maximum component by at least $\nu$, so it only takes at most $O(m/\nu)$ recursion depth for the volume of the largest component to be less than $m/2$. We stop the recursion once all the components are of volume at most half the original volume by just adding both sides of the minimum cut to $\mathcal{R}$.
\end{proof}

We analyze the total running time of \kcp{} by separating the calls to \blp{} local subroutine from the remaining part of the \kcp{}.

\begin{claim}\label{alg:kcp:localcalls}
\Cref{alg:kcp} invokes \blp{} at most $O(m)$ times during the running of the whole algorithm i.e. over all $O(m/\nu)$ recursion levels.
\end{claim}

\begin{proof}
Local subroutine \blp{} is always invoked on vertices that belong to $L$. So it is convenient to bound the number of vertices added to the candidate list during the whole algorithm course.

We initially start with the $L=V$, the whole vertex set of size $n$. When each $(\nu, k)$-set is separated from $\hat{G}$, we add the endpoints of the cut edges on the larger side to the candidate list. From $\Cref{lem:newcandidates}$, we know that any new extreme sets introduced due to removing cut edges should contain endpoints of cut edges. Hence it is enough to check for $(\nu,k)$-extreme sets just around the cut edges endpoints. As per the guarantee of \blp{} we find the $(\nu,k)$-sets with high probability if any $(\nu,k)$-extreme sets exist. Since we are removing the $(\nu,k)$-set from the graph, the extreme set is not present in the residual graph.

Once a cut edge is removed from the graph, it cannot be part of any of the sub-graphs present as a part of the later recursion. So every edge introduces a new vertex into list $L$, and no edge becomes a cut edge again. Hence a trivial bound on the number of vertices added to the candidate list is, at most, the number of edges. So total number of invocations of \blp{} subroutine is at most $m + n$ which is $O(m)$.
\end{proof}

Since number of invocations to \blp{} is $O(m)$, we can union bound the failure probability over all the invocations of \blp{} which is $\frac{O(m)}{n^c}$ still leading to a total error probability of at most $\frac{1}{\poly(n)}$ for sufficiently large constant $c$. Similarly, we can bound the error probability for $O(m/\nu)$ invocations of the global mincut algorithm. Since the sizes of the graphs during $O(m/\nu)$ invocations of the algorithm are at least $\Omega(n^{O(1)})$. If the error probability of global mincut algorithm is at most $\frac{1}{n^c}$ for some sufficiently large constant $c$ then error probability over all the $O(m/\nu)$ invocations is at most $\frac{1}{\poly(n)}$ as required. This proves that the error probability of \Cref{alg:kcp} is at most $\frac{1}{\poly(n)}$.

\begin{lemma}\label{alg:kcp:runtime}
\Cref{alg:kcp} runs in time at most $O(m^{1.75}\log^{2.75} n)$ for an $m$ edge $n$ vertex weighted graph.
\end{lemma}

\begin{proof}
From \Cref{alg:kcp:localcalls} we have bounded the number of calls to the \blp{} subroutine which is at most $O(m\nu^3 \log^2 n)$ from \Cref{thm:blp}.

At each recursion level of \kcp{}, we apply the global mincut algorithm on a graph with edges less than $m$ and vertices less than $n$. Hence we can trivially bound the run time due to the global mincut being at most $O(m/\nu)\cdot O(m\log^3 n)$.

So the total run time of the \kcp{} is at most $O(m\nu^3 \log^2 n + \frac{m^2 \log^3 n}{\nu})$. Optimizing with $\nu = O((m\log n)^{1/4})$ the run time of \kcp{} would be $O(m^{1.75}\log^{2.75} n)$.
\end{proof}

Now, we analyze the running time of \Cref{alg:mcp}.

\begin{lemma}\label{alg:mcp:runtime}
\Cref{alg:mcp} runs in time at most $O(m^{1.75} \log^{3.75} n)$ on a $m$ edge $n$ vertex weighted graph.
\end{lemma}

\begin{proof}
From the guarantee of the \Cref{alg:kcp}, we have that the volume of the graph induced on each set that belongs to $\mathcal{R}$ is half the original volume. Hence the recursion depth of the \mcp{} subroutine is at most $\log m = O(\log n)$.

Let us then bound the run time taken by \mcp{} at every recursion level. The volume of all the graphs induced on the partitioned sets is bounded by $m$. Let $m_i(<m/2), n_i$ be the number of edges and vertices in the $i^{th}$ graph induced on the partition of vertices after applying \kcp{} subroutine. Because these form a partition of vertices of the original graph and we only consider induced graphs on all these partitions $\sum_{i}n_i = n$, total volume is bounded by $m$, i.e. $\sum_{i}m_i \leq m$.

Hence we have the following recurrence equation.
\[T(m) = \sum_{i}T(m_i) + O(m^{1.75}\log^{2.75}n)\]
\[m_i \leq m/2 \text{ for all i}\]

Since $\sum_{i}O(m_i^{1.75}\log^{2.75}n_i) \leq O(m^{1.75}\log^{2.75}n)$. Hence it takes $O(m^{1.75}\log^{2.75}n)$ at each recursion level and there are at most $O(\log n)$ recursion level. Therefore the run time of the \mcp{} subroutine is at most $O(m^{1.75}\log^{3.75}n)$.
\end{proof}

\begin{proof}[Proof of \Cref{mcp:mainthm}]
From \Cref{lem:mcp:correctness} we have \Cref{alg:mcp} that computes the maximal $k$-edge-connected partition of the graph $G=(V,E,w)$ with parameter $k$ as input, which takes running time at most $O(m^{1.75}\log^{3.75}n)$ from \Cref{alg:mcp:runtime} proving \Cref{mcp:mainthm}.
\end{proof}

\section{Local algorithm for small cardinality cuts}\label{sec:localalg2}

The main goal of this section is to prove \Cref{thm:local main} (\ref{enu:local size}). In this section, we only care about $(\sigma,k)$-sets, small cardinality cuts of small cut size. However, they can have an arbitrarily high volume when the degree of the vertices is large. Hence applying \Cref{alg:blp} would lead to a large running time. Thus we overcome this by doing additional book-keeping to find $(\sigma,k)$-sets in $\Otil(\sigma^4)$ time.

We describe the data structure called sorted adjacency list for book-keeping in \Cref{sec:6.1} and then use it to speed up the \blp{} in \Cref{sec:6.2}.

\subsection{Sorted adjacency list}\label{sec:6.1}
Given a graph $G = (V,E,w)$ with a rank function $r:E\rightarrow [0,1]$ undergoing edge deletions we maintain adjacency lists of each vertex sorted according to rank. Along with a sorted adjacency list, we also maintain an adjacency matrix to check whether an edge is present or not in the graph in $O(1)$ time. We initially start with empty adjacency lists $\mathcal{L}$ and empty adjacency matrix $\mathcal{M}$, and we can modify our data structure using the following operations.

\begin{tcolorbox}
\label{SAL}
\begin{itemize}
    \item \ins($e=(u,v)$):
    \begin{enumerate}
        \item Add edge to the adjacency matrix $\mathcal{M}$.
        \item Insert edge to adjacency lists of $\mathcal{L}[u],\mathcal{L}[v]$. As the list is sorted according to rank and of size at most $\Delta$, it takes at most $O(\log \Delta)$ time.
    \end{enumerate}
    \item \del($e=(u,v)$):
    \begin{enumerate}
        \item Mark the edge deleted in adjacency matrix $\mathcal{M}$.
        \item We can store the pointer to edge $e$ in the adjacency matrix thus we can find it in $O(1)$ time and can also delete in adjacency lists $\mathcal{L}[u],\mathcal{L}[v]$ in $O(1)$ time if we maintain $\mathcal{L}[u],\mathcal{L}[v]$ as doubly linked lists.
    \end{enumerate}
    \item \nextedge($X$):
    \begin{enumerate}
        \item Return $\min \{x\in X : \min(\mathcal{L}[x])\}$ in $O(|X|)$ time.
    \end{enumerate}
\end{itemize}
\end{tcolorbox}

We construct the sorted adjacency list of the graph $G$ with respect to a rank function $r$ by inserting (\ins) each edge into the data structure, which takes a total run time of at most $O(m\log \Delta) = O(m\log n)$.

As seen in \blp{} to boost the probability of finding the extreme sets, we repeat the \lp{} sub-routine $O(\sigma^2 \log n)$ times with independent rank functions. So we will need to individually maintain the above data structure for all the $O(\sigma^2 \log n)$ rank functions.

It takes a total time $O(m \log n)\cdot O(\sigma^2 \log n) = O(m\sigma^2 \log^2 n)$ to construct the data structure with respect to all rank functions. In the next subsection, we see how to use the data structure to speed up \lp{}.

\subsection{Speedup \lp{}}\label{sec:6.2}
We will formalize how to use the sorted adjacency list below to improve the run time of \lp{}. We omit the volume constraint $\nu$ while expanding the set $X$ in the algorithm, as we only care about small cardinality cuts. The rest of the algorithm remains the same.

\begin{algorithm}[H]
\KwData{$G = (V,E,w),x,\sigma,k,r$ where $r$ is a rank function $r:E\rightarrow [0,1]$}
\KwResult{$(x,\sigma,k)$-set respecting $r$ if exists else $\bot$}
$X = \{x\}$\;
\While{$w(\delta(X)) \geq k$ \sout{and $vol(X) < \nu$} and $|X| < \sigma$ \label{alg:lpimproved:3}}
{
    Find the edge with minimum rank, $e = (u,v) \in E(X,V\setminus X) = \delta(X)$\;\label{alg:lpimproved:4}
    Update $X = X \cup \{v\}$\;
}
\eIf{$w(\delta(X)) < k$}
{\Return{$X$}}
{\Return{$\bot$}}
\caption{\lp}
\label{alg:lpimproved}
\end{algorithm}

In \Cref{improved:localprimruntime}, we analyze the running time of \lp{} that uses the sorted adjacency list.

\begin{claim}\label{improved:localprimruntime}
\lp{} takes at most $O(\sigma^2)$ time using sorted adjacency list.
\end{claim}

\begin{proof}
In each iteration of the while loop, we find the next edge to expand the set $X$ along using \nextedge{}, which takes $O(|X|)$ time. Let $v$ be the endpoint of the edge that is not in $X$. To maintain $w(\delta(X))$ we need to add the $\deg(v)$ (weighted degree) which takes at most $O(1)$ time and delete the weight of the edges that are incident on $v$ from vertices currently present in $X$. So it takes at most $O(|X|)$ to compute the cut size after contracting $v$.

After expanding the set by including the vertex $v$, we need to delete the edges whose both endpoints are in $X$ as the next edge to expand set $X$ should not be from the internal edges of set $X$. After expanding set $X$ with $v$, there are at most $|X|$ many edges that are becoming internal edges (edges whose both endpoints lie in $X$) and need to be deleted (\del) from the data structure. So it takes at most $O(|X|)$ time.

So each iteration of the while loop takes $O(|X|)$ time. As $|X|$ ranges from $1$ to at most $\sigma$ the total run time is at most $O(\sigma^2)$.
\end{proof}

Note that we are modifying the sorted adjacency list data structure for a particular rank function $r$. So we maintain a journal of the edits made to the data structure, which are at most $O(\sigma^2)$ as we only delete the internal edges. Thus we can undo the changes once we have completed the sub-routine \lp. So it still takes $O(\sigma^2)$ time to run the sub-routine and maintains the data structure intact.

Similar to \lp{}, we modify \blp{} to find small cardinality local cuts of small cut size. Note that the \lp we use here is the one for finding small cardinality cuts.

\begin{algorithm}[H]
\KwData{$G = (V,E,w),x,\sigma,k$}
\KwResult{$(x,\sigma,k)$-set or $\bot$ if no $(x,\sigma,k)$-extreme set exists}
\For{$ i \in [O(\sigma^2 \log n)]$}
{
    $X = \lp(G,x,\sigma,k,r^i)$\;
    \If{$X \neq \bot$}
    {
        \Return $X$\;
    }
}
\Return $\bot$\;
\caption{\blp}
\label{alg:blpimproved}
\end{algorithm}

From \Cref{lem:blp:runtime} we can conclude the following \Cref{thm:blpimproved}.

\begin{theorem}
\label{thm:blpimproved}
There exists a randomized algorithm given a graph $G = (V,E,w)$, a vertex $x$ and parameters $\sigma,k$ finds either a $(x,\sigma,k)$-set or guarantees that no $(x,\sigma,k)$-extreme set exists with high probability in time $O(\sigma^4\log n)$.
\end{theorem}

\begin{proof}
Since the probability of finding the $(\sigma,k)$-set using \lp{} is only dependent on number of vertices in the set and it is at least $\frac{2}{\sigma(\sigma-1)}$ as seen in \Cref{lem:rrrprob}. Hence we can repeat the \lp{} subroutine $O(\sigma^2 \log n)$ times to improve the probability thus taking a total run time of $O(\sigma^4 \log n)$.
\end{proof}

This proves the \Cref{thm:local main}(\ref{enu:local size}).

\section{\texorpdfstring{$\Tilde{O}(mn^{0.8})$}{mn0.8} algorithm for maximal \texorpdfstring{$k$}{k}-connected partitions}\label{sec:paramchange}

We have seen in the \Cref{sec:recdepth} how a local subroutine \blp{} is used to improve the recursion depth and thus improve the run time, at least when the graph is sparse. However, $\Tilde{O}(m^{1.75})$ is still high when $m$ is in dense regime. Let us see an improved version of the above algorithm below by using the fast \blp{} from \Cref{sec:localalg2} and changing the parameter to $\sigma$ instead of $\nu$, with the general framework of the algorithm remaining intact. This results in the second running time claimed in \Cref{thm:main}. The main result that we prove in this section is the following theorem.

\begin{theorem}\label{thm:mcpimproved}
There exist a randomized algorithm that takes a graph $G=(V, E,w)$ with $m$ edges and $n$ vertices and finds the maximal $k$-edge-connected partition with high probability in at most $O(mn^{0.8}\log^{3.6} n)$ time.
\end{theorem}

\subsection{Parametrize over \texorpdfstring{$n$}{n}}
In the previous local subroutine, we have parameterized over the volume of the local cut $\nu$, thus resulting in at most $O(m/\nu)$ recursion depth. When we parameterize over $\sigma$, we can change the recursion depth, which would be at most $O(n/\sigma)$. We formally prove it in the following lemma.

\begin{lemma}
Recursion depth of the \kcp{} using updated \blp, \lp{} subroutine is at most $n/\sigma$.
\end{lemma}

\begin{proof}
At each level, when all the $(\sigma,k)$-extreme sets are removed and if the residual graph is still not $k$-edge-connected, then the minimum cut found using Karger's \mcut{} has at least $\sigma$ vertices on both sides. If not, we would find a $(\sigma,k)$-extreme set contradicting (with high probability) that no $(\sigma,k)$-extreme set exists. Since the minimum cut is a $\sigma$-balanced cut, the number of vertices in the large component reduces by at least $\sigma$. So, it takes at most $O(n/\sigma)$ recursion depth to reduce the number of vertices in the largest component to half its original size, i.e., $n/2$. Hence the recursion depth of \kcp{} is at most $O(n/\sigma)$.
\end{proof}

\subsection{Run time analysis}
The analysis is very similar to that of the previous algorithm.

\begin{lemma}
The run time of \kcp{} is at most $O(mn^{0.8}\log^{2.6}n)$.
\end{lemma}

\begin{proof}
As proven in \Cref{alg:kcp:localcalls} we invoke \blp{} at most $O(m)$ times, each of which takes $O(\sigma^4 \log n)$ time. So it takes $O(m\sigma^4 \log n)$ time for the local subroutine calls.

The recursion depth of \kcp{} is at most $O(n/\sigma)$ and the run time required for applying Karger's \mcut{} at each recursive level is at most $O(m\log^3 n)$.

We also need to account for the time taken to construct the data structure for all the $O(\sigma^2\log n)$ rank functions which amount to at most $O(m\sigma^2 \log^2 n)$.

Hence the total runtime is $O(m\sigma^2\log^2 n) + O(m\sigma^4\log n) + O(m\log^3 n)\cdot O(\frac{n}{\sigma})$. Optimizing with $\sigma = O((n\log^2n)^{1/5})$ the total run time is at most $O(mn^{0.8}\log^{2.6}n)$.
\end{proof}

Since \mcp{} calls \kcp{} which breaks the graph into partition in which each set is of size at most half of the original size. Hence recursion depth of \mcp{} is at most $\log n$. With similar argument stated in \Cref{alg:mcp:runtime} we conclude that the run time taken to find maximal $k$-edge-connected partition is at most $O(mn^{0.8}\log^{3.6}n)$.

This proves the \Cref{thm:main}. Note that $\Tilde{O}(mn^{4/5})$ is an improvement over $\Tilde{O}(m^{1.75})$ in the dense regime where $m = \Omega(n^{16/15})$.

\section{Conditional lower bounds}\label{sec:lb}

As we have seen from the above sections, the main bottleneck in our framework to find maximal $k$-edge-connected partitions is the sub-routine to find a set of cut size less than $k$ having small cardinality or volume. A natural question is ``Can we improve the running time dependence on parameters $\nu,\sigma$ for these problems?". In this section, we give lower bounds for problems \localmst and \minextset. We use the conjectured hardness of the OMv problem to give running time lower bounds for these problems.

We state the OMv conjecture below. Let $M$ be a boolean matrix of dimension $n\times n$. We have to compute the matrix-vector product of $v_1,\cdots v_n$ one after another in an online fashion. The naive running time to compute the vector product of all the vectors would be $O(n^3)$. It is conjectured that even after allowing a polynomial amount of pre-processing time, one cannot polynomially improve over the naive running time.%

\begin{conjecture}[OMv Conjecture]\label{omv}
 There is no algorithm that computes matrix-vector products of $n$ vectors coming in an online fashion in a total time of $n^{3-\Omega(1)}$ with an error probability of at most $1/3$.%
\end{conjecture}

In particular, we use the hardness of the independent set query problem proved in \cite{henzinger2015unifying} to give lower bounds to \Cref{prob:minextset} and \Cref{prob:localmst}. 
More precisely, By combining Theorems 2.7 and 2.12 from \cite{henzinger2015unifying}, we obtain the following hardness result:
\begin{theorem}[\cite{henzinger2015unifying}]\label{lb:indset}
Assuming \Cref{omv}, there is no algorithm that computes whether any subset, $S \subset V$ is independent or not in the graph $G=(V, E, w)$ with error probability at most $1/3$ in $O(|S|^{2-\Omega(1)})$ time even after allowing a polynomial pre-processing on the graph $G$.
\end{theorem}

The rest of the section is organized as follows. In \Cref{sec:8.1} we prove the lower bound for the cardinality-constrained small cut containing $x$, $(x,\sigma,k)$-\minextset problem. In \Cref{sec:8.2} we show that even the simple problem of cardinality-constrained small cut containing $x$ and respecting a given rank function $r$, $(x,\sigma,k)$-\localmst cannot have a faster algorithm when only adjacency list and matrix of graph $G$ are given. Finally, in \Cref{sec:8.3} we extend the NP-hardness reduction given in \cite{fomin2013parameterized} to give $W[1]$-hardness albeit in the weighted case.

\subsection{Lower bound for \minextset}\label{sec:8.1}

Recall the $(x,\nu,\sigma,k)$-\minextset{} problem.

\localkcutprob*

We have given two different local algorithms in \Cref{sec:extsets} and \Cref{sec:localalg2} respectively for different parameter variations of \minextset that run in time $\Tilde{O}(\nu\sigma^2)$ and $\Tilde{O}(\sigma^4)$. We would prove a conditional lower bound for $(x,\sigma,k)$-\minextset as stated below.

\begin{theorem}\label{lb:minextset}
Assuming \Cref{omv}, there is no algorithm that, given access to the adjacency lists and matrix of $G$ (with no further preprocessing) solves $(x,\sigma,k)$-\minextset{} in $O(\sigma^{2-\Omega(1)})$ time.
\end{theorem}

Before proving the above theorem, we will first show the reduction that connects the problem of finding the $(\sigma,k)$-extreme set to the problem of deducing whether a set is independent.

Given a graph $G=(V, E,w)$ and a set $S$ of size $\sigma$ we can find the volume of $S$ in $O(\sigma)$ time. Let $k=\vol_G(S),\sigma = |S|$. Define a new graph $G_S=(V, E',w')$, which is an exact copy of $G$ but also has an additional weighted cycle consisting of all the vertices in $S$ with each edge of weight $k+1$. If an edge of the newly added cycle already exists in $E$, then just increment the weight of that edge by $k+1$ in $w'$ to keep $G_S$ simple graph. We replace the infinite value in the introduction with $k+1$ as it is enough to get the desired result. We prove the following equivalence:

\begin{figure}[htbp!]
    \begin{center}
        \includegraphics[scale=0.3]{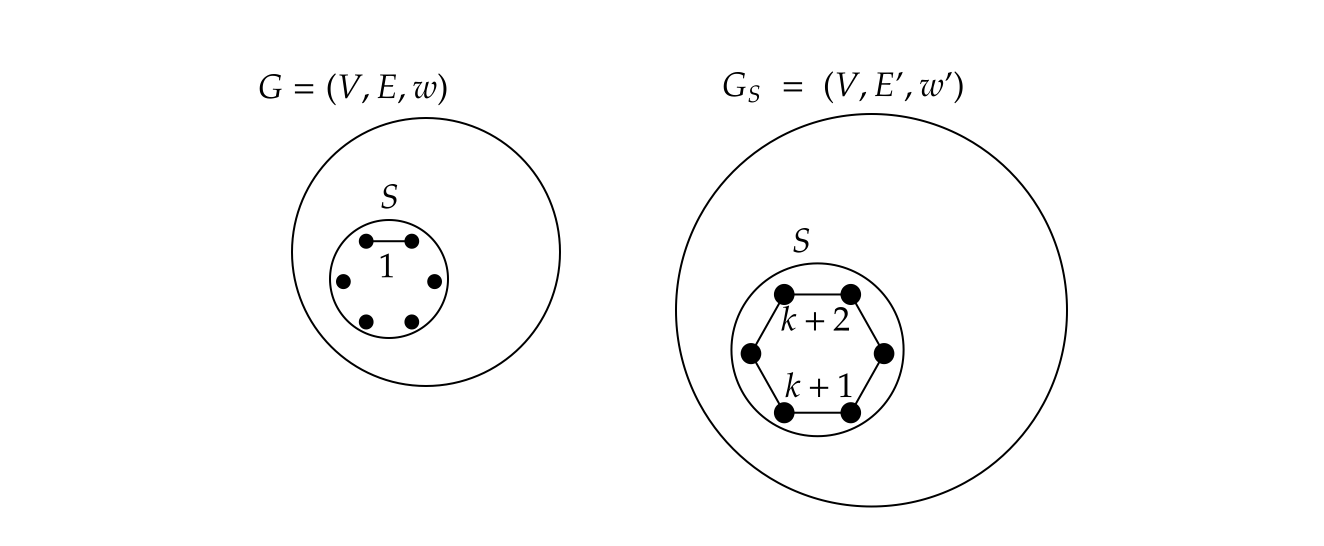}
    \end{center}
    \caption{Lowerbound for $(x,\sigma,k)$-\minextset}
\end{figure}

\begin{lemma}\label{redn:minextset}
$E(S,S)\neq \phi$ in $G$ iff $S$ is $(\sigma+1,k)$-extreme set in $G_S$.
\end{lemma}

\begin{proof}
($\impliedby$) If $S$ is $(\sigma+1,k)$-extreme set in $G_S$ then we have $\delta_{G_S}(S)<k$ from the definition of $(\sigma+1,k)$-extreme set. According to our construction of $G_S$ we have $\delta_G(S) = \delta_{G_S}(S)$. Hence, $\delta_G(S) = \delta_{G_S}(S) < k = \vol_G(S)$. Since, $\delta_G(S) < \vol_G(S)$ we have $E(S,S)\neq \phi$.

($\implies$) If $E(S,S)\neq \phi$ then $\delta_G(S)<k$ and $|S|<\sigma+1$, hence $S$ is a $(\sigma+1,k)$-set. $S$ is also an extreme set in $G_S$ as it is unique minimum cut of $G_S/\overline{S}$ as any other cut of $G_S/\overline{S}$ would consist an edge from the cycle of weight $k+1$ and hence the weight of the cut is more than $k$.
\end{proof}

\begin{obs}\label{onlyS}
For any $x\in S$, $S$ is the only $(x,\sigma+1,k)$-set in $G_S$ which is also extreme.
\end{obs}

\Cref{onlyS} follows from the above construction of $G_S$. Now, we prove the lower bound for the problem of $(x,\sigma,k)$-\minextset.

\begin{proof}[Proof of \Cref{lb:minextset}]
Assume for contradiction that there is an algorithm $\mathcal{A}$ that solves $(x,\sigma,k)$-\minextset and runs in $O(\sigma^{2-\eps})$ time for some constant $\eps>0$. We can use this algorithm to find whether a given set $S$ is independent or not in running time $O(\sigma^{2-\eps})$ as follows.

Given a set $S$ to check whether it is independent or not in graph $G$, construct the graph $G_S$ as described above in $O(\sigma)$ time. Choose an arbitrary vertex $x \in S$ and use $\mathcal{A}$ to find $(\sigma+1,k)$-set containing $x$ in $G_S$. According to the runtime guarantee of the algorithm, this takes at most $O((\sigma+1)^{2-\eps}) = O(\sigma^{2-\eps})$ and finds whether there is a $(x,\sigma+1,k)$-set or return $\bot$ if no $(x,\sigma+1,k)$-extreme set exists.

If $\mathcal{A}$ returns a $(x,\sigma+1,k)$-set with any $x\in S$ as input, then it has to be $S$ from \Cref{onlyS}. If $S$ is the $(\sigma+1,k)$-set then $S$ is not independent.

If $\mathcal{A}$ returns $\bot$, then it implies that there is no $(x,\sigma+1,k)$-extreme set in $G_S$. From \Cref{onlyS}, the only $(x,\sigma+1,k)$-set is $S$ which is also extreme. Since, $\mathcal{A}$ implies no $(x,\sigma+1,k)$-extreme set, there is no $(x,\sigma+1,k)$-set which implies $S$ is not a $(x,\sigma+1,k)$-set. Since $S$ contains $x$ and has size less than $\sigma+1$, the only condition that violates is the cut size constraint, so $\delta_G(S) = \delta_{G_S}(S)\geq k = \vol_{G}(S)$. Hence $S$ is independent.

So, we can compute whether $S$ is independent or not in $O(\sigma^{2-\eps})$ time using $\mathcal{A}$, which contradicts \Cref{omv}. Hence, our assumption that there exists a $O(\sigma^{2-\eps})$ algorithm for \Cref{prob:minextset} is false for any constant $\eps>0$.
\end{proof}

Although \Cref{lb:minextset} rules out any polynomially faster algorithm than $O(\sigma^2)$ for any general $\sigma$. It still does not rule out an algorithm that only works when the size of the set $S$ is small compared to the vertex set $V$. This is important because according to our parameters setting $\sigma=\Otil(n^{1/5})$. So we have to rule out a better algorithm for \Cref{prob:minextset} even in this smaller parameter regime. Below we give a strong lower bound, which rules out an algorithm that performs polynomially better than $O(\sigma^2)$, even when we are guaranteed that the query set's size, $|S| = \sigma \ll n$.

\begin{theorem}\label{lb:restrictparam}
Assuming \Cref{omv}, there is no algorithm that, given access to adjacency list and matrix of $G$ (with no further preprocessing) solves \Cref{prob:minextset} in $O(\sigma^{2-\Omega(1)})$ even when $\sigma \leq n^{\gamma}$ for every constant $\gamma\in[0,1)$.
\end{theorem}

\begin{proof}
Let $\mathcal{A}$ be an algorithm that solves \Cref{prob:minextset} in $O(\sigma^{2-\eps})$ for some constant $\eps>0$ when $\sigma \leq n^{\gamma}$. We can create an algorithm that solves independent set query problem polynomially faster than $O(|S|^{2})$ for set $S$ of vertices in any graph of size $n$.

Given a graph $G=(V,E)$, construct an arbitrary graph $H = (V_H,E_H)$ of vertex set size $|V_H|=n^{1/\gamma}$ and embed the graph $G$ on some arbitrarily chosen subset of $n$ vertices in $H$ which is of size at most $|V_H|^\gamma$.

Given any set $S\subseteq V$ in graph $G$ construct the corresponding subset $S_H$ of vertices in $H$ on which graph $G$ is embedded and add the cycle along the vertices $S_H$ with weight $\vol_G(S)+1$ as constructed in \Cref{lb:minextset}. We can now use $\mathcal{A}$ for finding $(\sigma+1,k)$-extreme set in $H$ because the queried set $S$ has size $|S| \leq n = |V_H|^{\gamma}$.
So as proved in \Cref{lb:minextset} we can find whether the set $S_H$ is an independent set or not in $H$ which is the same as finding whether $S$ is independent in $G$ in run time $O(|S|^{2-\eps})$ contradicting \Cref{omv}.
\end{proof}

\subsection{Lower bound for \localmst}\label{sec:8.2}
Recall the problem of \localmst{}. In this subsection, we give a lower bound for this problem.

\localmstprob*

In \Cref{sec:extsets} we have seen an algorithm for solving \Cref{prob:localmst} which finds a $(\nu,\sigma,k)$-set that respects the rank function in $O(\nu\log \nu)$ time. In \Cref{sec:localalg2} we have seen another algorithm that finds a $(\sigma,k)$-set that respects a given rank function in $O(\sigma^2)$ time. Below we will show that even this simple \Cref{prob:localmst} of finding a $(x,\sigma,k)$-set that respects the given rank function $r$ cannot be solved in run time polynomially better than $O(\sigma^2)$.

Again as in the previous lower bound, we will show that no algorithm solves \Cref{prob:localmst} polynomially faster than $O(\sigma^2)$ for any parameter range $\sigma$. We can use the same argument as in \Cref{lb:restrictparam} for proving the lower bound for any algorithm that only works in the smaller range of parameter $\sigma$. We state the theorem below.

\begin{theorem}\label{lb:localmst}
Assuming \Cref{omv}, given a graph $G= (V,E,w)$ with adjacency lists and matrix, parameters $\sigma,k$ a vertex $x\in V$ and a rank function $r$, there is no algorithm that solves \Cref{prob:localmst} in time $O(\sigma^{2-\Omega(1)})$.
\end{theorem}

Given a graph $G = (V, E,w)$ and a query set $S$, we create $G_S$ by adding a cycle of weight $k+1$, where $k=\vol_G(S)$ as in the lower bound for \minextset. We define a rank function $r_S$ with rank value $0$ on the edges added as part of the cycle, and all other edges have rank $1$. From the construction, it is clear that $S$ respects the rank function $r_S$ in $G_S$ as it has a spanning tree in $G_S[S]$ where each edge has rank value $0$ and any edge that belongs to cut $E(S, V\setminus S)$ has rank value $1$ which is more than $0$.

\begin{proof}[Proof of \Cref{lb:localmst}]
Assume we have an algorithm $\mathcal{A}$ that solves \Cref{prob:localmst} in time $O(\sigma^{2-\eps})$ for some $\eps > 0$. Given any graph $G$ we do no pre-processing and whenever we get a query set $S$ we construct $G_S$ and the corresponding rank function $r_S$ (implicitly).

Hence we can use the algorithm $\mathcal{A}$ on graph $G_S$ with any arbitrary $x\in S$ and rank function $r_S$ to find a $(x,\sigma+1,k)$-set in $O(\sigma^{2-\eps})$ time, where $\sigma=|S|$ and $k=\vol_G(S)$.

If $\mathcal{A}$ finds a $(x,\sigma+1,k)$-set respecting $r_S$ in $G_S$ then it has to be $S$, because from \Cref{onlyS} we know that $S$ is the only $(x,\sigma+1,k)$-set. Hence $S$ is not an independent set.

Else $\mathcal{A}$ finds no set and confirms that $S$ is independent. It only takes $O(|S|^{2-\eps})$ time thus contradicting \Cref{omv}. So no such algorithm $\mathcal{A}$ exists.
\end{proof}

\subsection{Lower bound for SmallLocalSeparation}\label{sec:8.3}

\cite{fomin2013parameterized} proves \Cref{prob:smalllocsep} is NP-complete when the graph is unweighted by giving a reduction from \textsc{Clique} problem in regular graphs. We extend the theorem and prove that \Cref{prob:smalllocsep} is $W[1]$-hard for weighted graphs when parameterized by $\sigma$.
It is important to note that \Cref{prob:smalllocsep} has a stronger guarantee than the \Cref{prob:minextset}. We call it \textsc{SmallLocalSeparation}.

\begin{problem}[\textsc{SmallLocalSeparation}]\label{prob:smalllocsep}
Given a graph $G=(V, E,w)$, parameters $\sigma,k$ and a vertex $x \in V$, either find a $(x,\sigma,k)$-set or guarantee that no such set exists.
\end{problem}

It is worthwhile to note that the \minextset problem is a relaxation of the \textsc{SmallLocalSeparation}. Although in both problems we find a $(x,\sigma,k)$-set, the guarantee that we give when we do not find such a set is weaker in \minextset as we are only ruling out $(x,\sigma,k)$-extreme sets, while stronger in \textsc{SmallLocalSeparation} as we have to rule out all $(x,\sigma,k)$-sets. This relaxation is crucial to the tractability of \minextset as we exploit the structure involved in extreme sets.

\begin{lemma}\label{lb:smalllocsep}
\Cref{prob:smalllocsep} is $W[1]$-hard in parameter $\sigma$.
\end{lemma}

At a high level, the proof directly extends from the NP-complete proof of \Cref{prob:smalllocsep} in the unweighted case in \cite{fomin2013parameterized}, by replacing the base clique of $dn$ nodes with a single vertex and joining it to all the vertices of the original graph directly with edges of weight $d$. The base clique is needed in their construction to keep the graph unweighted, however, we chose to prove for the weighted case at the cost of getting a better hardness i.e., proving $W[1]$-hard.

We know that $t$-\textsc{Clique} problem is $W[1]$-hard when parameterized by $t$, where $t$ is the number of vertices in the clique, even on $d$-regular graphs.

We define the construction of graph $G'$ from $G$ that reduces $t$-\textsc{Clique} problem to \Cref{prob:smalllocsep}, that is used in proving \Cref{lb:smalllocsep}.

Let a $d$-regular unweighted graph $G = (V, E)$ be the input of the $k$-Clique problem. The input for \textsc{SmallLocalSeparation} $G' = (V', E', w')$ is constructed as follows. Vertex set $V'$ consists of the vertex $x$ to be used as a terminal in \textsc{SmallLocalSeparation}, so we call it ``terminal node" and set of nodes corresponding to the vertices and edges of graph $G$, call them ``vertex nodes" and ``edge nodes" respectively. Hence $|V'| = 1+n+m$. Add an edge of weight $d$ from $x$ to all vertex nodes. Hence the degree of $x$ would be $dn$. For every edge $e=(u,v)\in E$, add two edges, each of weight one, from the edge node to the corresponding vertex nodes in $V'$. Below we prove that the graph $G'$ consists of a small local cut of size $\sigma$ if and only if $G$ has a $t$-clique and $\sigma = \poly(t)$ as required for proving fixed-parameter intractability.

\begin{figure}[htbp!]
    \begin{center}
        \includegraphics[scale=0.3]{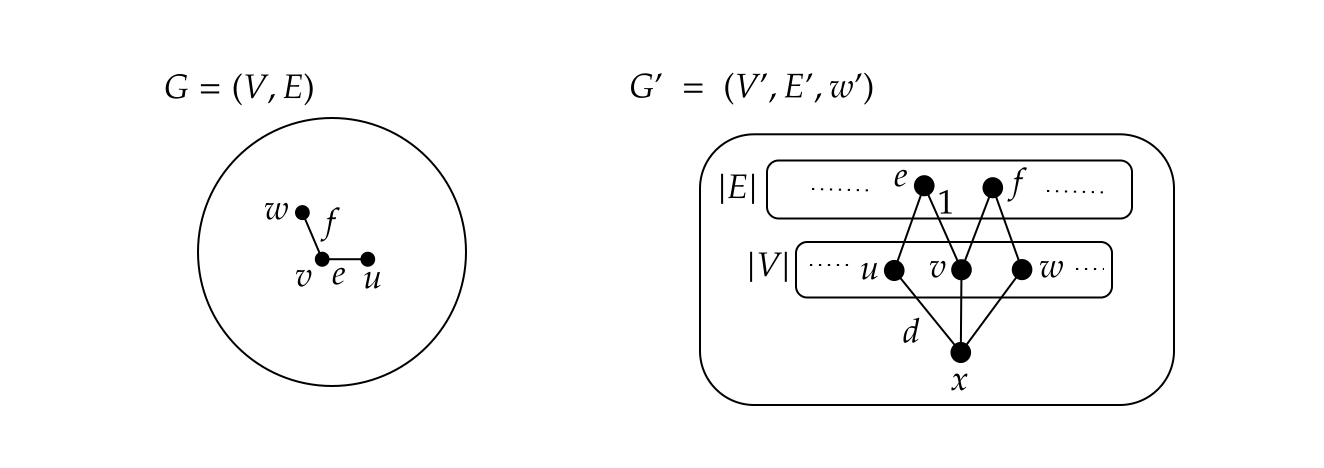}
    \end{center}
    \caption{$W[1]$-hardness of \textsc{SmallLocalSeparation}}
\end{figure}

\begin{lemma}\label{redn:smalllocsep}
Let $\sigma = 2+t+\binom{t}{2} = \poly(t)$ and $k=dn-2\binom{t}{2}+1$. $(x,\sigma,k)$-set exists in $G'$ iff a clique of size $t$ exists in $G$.
\end{lemma}

\begin{proof}

($\impliedby$) Let $S$ be a clique of size $t$ in $G$. Construct the set $S'$ as follows. It consists of the terminal node $x$, vertex nodes corresponding to the set $S$ in $V'$ and all the edge nodes corresponding to the edges $E(S,S)$. Since $S$ is a clique in $G$, the size of $S'$ is $1+|S|+|E(S,S)| = 1+t+\binom{t}{2} < \sigma$.

The size of the cut $\delta_{G'}(S')$ is the total number of edges going out of $S'$. The weight of the cut edges incident on $x$ is $d(n-|S|) = d(n-t)$. The degree of vertex nodes excluding the edges to the terminal node is $d$ according to our construction. Hence, the total degree of vertex nodes is $d|S|$. However, since all the edge nodes $E(S, S)$ are included in the set $S'$, the degree to them is subtracted. Since each such edge is counted twice because of both endpoints. The weight of the cut edges due to vertex nodes is $d|S|-2|E(S, S)| = dt-2\binom{t}{2}$. Since all the edges from edge nodes are inside the set $S'$ they do not contribute to cut edges. Hence, the cut size of $S'$ is $d(n-t)+dt-2\binom{t}{2} = dn-2\binom{t}{2} < k$. This proves that $S'$ is a $(x,\sigma,k)$-set.

($\implies$) If there exists a $(x,\sigma,k)$-set $S'$ in $G'$, then it contains $x$ by definition. Let $V_o, E_o$ be the vertex nodes and edge nodes present in the set $S'$. $V_o$ cannot be empty as every edge node adds a degree $2$ and the cut size would be $dn+2|E_o|$, which violates the cut size constraint. Hence, $V_o\neq \emptyset$. Assuming $E_o=\emptyset$. The cut size would be $|\delta_{G'}(V_o \cup \{x\})| = d(n-|V_o|)+d|V_o| = dn$ which violates the cut size constraint. Hence $E_o\neq \emptyset$.

There are three types of edges in $G$ with respect to the vertex set $V_o$ chosen, internal edges $E_G(V_o,V_o)$, cut edges $E_G(V_o,V\setminus V_o)$ and external edges $E_G(V\setminus V_o,V\setminus V_o)$. (Note that we misuse the notation $V_o$, what we meant is the corresponding vertex set in $G$.)

Adding an internal edge node to $S'$ would reduce the cut size by two, a cut edge node to $S'$ would not change the cut size, and an external edge node to $S'$ would increase the cut size by two. Hence to reduce the cut size from $dn$ to below $k = dn-2\binom{t}{2}+1$ we should add at least $\binom{t}{2}$ internal edges. To have $|E_G(V_o,V_o)|\geq \binom{t}{2}$, the size of $V_o$ must be at least $t$. Since the number of nodes in the set $S'$ is bounded by $1+t+\binom{t}{2}$, the number of vertices in $V_o$ cannot exceed $t$ making it equal to $t$. It implies $|E(V_o,V_o)| = \binom{t}{2}$. Thus the set of vertices corresponding to $V_o$ in $G$ have to induce a clique of size $t$.
\end{proof}

Given the above reduction, we can now prove the $W[1]$-hardness of the \textsc{SmallLocalSeparation}. Observe that the set $S'$ is not extreme in $G'$ as the size of degree cut of any edge node is $2$ which is less than the size of the cut $\delta_{G'}(S')$. Hence, this hardness does not work for \Cref{prob:minextset}.

\begin{proof}[Proof of \Cref{lb:smalllocsep}]
Given an instance of $t$-\textsc{Clique}, $d$-regular unweighted graph $G = (V,E)$ and a parameter $t$, we construct an instance of weighted graph $G' = (V',E',w')$ as shown above and set parameters $x,\sigma=\poly(t),k$.

From \Cref{redn:smalllocsep} any FPT algorithm in parameter $\sigma$ for \Cref{prob:smalllocsep} can be turned into an FPT algorithm for $t$-\textsc{Clique} as $\sigma = \poly(t)$. Since $t$-\textsc{Clique} is $W[1]$-hard, \Cref{prob:smalllocsep} is $W[1]$-hard with parameter as $\sigma$.
\end{proof}
\section{Approximating edge strength}\label{sec:apprxes}
In this section we prove the \Cref{cor:apprxedgstn}.

\begin{definition}[Edge Strength]
The strength of an edge $e=(u,v)$ in a graph $G$ is the largest $k$ for which there exists a set $S\subseteq V$ containing $u,v$ such that $G[S]$ is $k$-edge-connected.
\end{definition}

Hence, if two vertices of an edge belong to the same set in the maximal $k$-edge-connected partition then the edge has strength at least $k$. As we know, the maximum possible strength of any edge is $\Delta \leq (n-1)W = \poly(n)$ where $W$ is the maximum weight of an edge in the graph. As any graph can be disconnected by removing $n-1$ edges.
Assuming that $W$ is bounded by $\poly(n)$, edge strength is also bounded by $\poly(n)$.

\begin{lemma}\label{edgestrength}
Given a maximal $k$-edge-connected partition of vertices:
\begin{enumerate}
    \item Strength of edge contained in a subgraph is at least $k$.
    \item Strength of edge that is across two subgraphs is less than $k$.
\end{enumerate}
\end{lemma}

\begin{proof}
Let $\mathcal{R} = \{V_1,V_2,\cdots\}$ be the maximal $k$-edge-connected partition of vertices. First part trivially follows from the definition.

Any edge $e=(u,v)$ across two partitions cannot have a $k$-edge-connected component surrounding it. Assume for contradiction that a subgraph $H$ exists which is $k$-edge-connected and containing $u,v$ both. 

Since the edge is across two components in the maximal $k$-edge-connected partition, let $V_i$ be the component that contains the vertex $u$. Since $H$ is $k$-edge-connected and $G[V_i]$ is $k$-edge-connected, hence $H \cup G[V_i]$ is also $k$-edge-connected which contradicts the maximality of the $k$-edge-connected subgraph induced on $V_i$. Thus such a subgraph $H$ does not exist. Hence the strength of any inter-partition edge is less than $k$, proving the second part of the lemma.
\end{proof}

Hence we approximate the strength of every edge followed by upper and lower bounding it in an interval of $(1+\eps)$.

\begin{proof}[Proof of \Cref{cor:apprxedgstn}:]
The maximum value of $k$ for which a graph can be $k$-edge-connected is bounded by $\poly(n)$. Hence divide the space in powers of $(1+\eps)$ leading to $O(\frac{1}{\eps}\log n)$ many values.

We invoke the subroutine given in $\Cref{alg:mcp}$ with $k$ taking all powers of $(1+\eps)$. Using \Cref{edgestrength}, we can find upper and lower bounds to the strength of each edge up to a factor of $(1+\eps)$. 

So the total runtime is at most $\Otil(m\cdot\min(m^{3/4},n^{4/5})/\eps)$.
\end{proof}

\section*{Acknowledgements}
We thank Bundit Laekhanukit for insightful discussion at the early stage of this project.

\appendix

\section{Rank function and its guarantees}\label{sec:invprobdist}

\begin{proof}[Proof of \Cref{lem:rankfn}]
Let us start with a warm up first. Let $e_1,e_2,e_3$ be three edges that sample values $x,y,z$ in $[0,1]$. Then the probability of $x$ being minimum among the three values $x,y,z$ is
\[\Pr(x=\min(x,y,z)) = \int_{0}^{1} \int_{x}^{1} \int_{x}^{1} \,dz \,dy \,dx = \frac{1}{3}\]

Hence choosing the edge with a minimum rank among $x,y,z$ is the same as choosing a uniform edge. Now let us see how to extend this to edges with weights.

Since weights are integral, each edge now chooses as many samples as its weight, and we choose the edge with the minimum random value as in warm-up. So similar to the warm-up, the probability of a particular sample being minimum is

\[\text{Probability of any sample value is minimum} = \frac{1}{\sum_{f}w(f)}\]

However, since each edge has multiple samples, we choose an edge $e$ if any of its $w(e)$ samples of the edge are minimum. Hence the probability of choosing an edge $e$ is

\[\text{Probability of any $w(e)$ samples being minimum} = \frac{w(e)}{\sum_{f}w(f)}\].

It is the same as choosing an edge with a weighted probability. Note that the number of random samples needed to sample is $\sum_{f}w(f)$, which can be very large when the weights are large.

We can overcome this problem by simulating the random distribution of the minimum of multiple random samples by sampling a single random variable. Let $X = \min(S_1,S_2,\cdots,S_w)$ be the random variable that is needed to be computed for an edge with weight $w$ where $S_i$ are uniform random samples in $[0,1]$. The cumulative density function of $X$ is as follows.
\begin{align*}
    \Pr(X>t) &= \Pr(S_1>t \text{ and } S_2>t\text{ and }\cdots S_w>t)\\
    &= (1-t)^w\\
    \Pr(X\leq t) &= 1-\Pr(X>t)\\
    &= 1 - (1-t)^w\\
\end{align*}

We can simulate the distribution of $X$ for an edge with weight $w$ as follows. Let $R$ be a uniform random variable in $[0,1]$. $\Pr(R\leq p) = p$ for any $p\in [0,1]$. Defining $Y= f(R) = 1-(1-R)^{1/w}$, we have
\begin{align*}
    \Pr(R\leq p) &= p = \Pr(Y \leq 1-(1-p)^{1/w}) \tag{as $Y= f(R)$}\\
    \Pr(Y \leq t) &= 1-(1-t)^w \tag{taking $1-(1-p)^{1/w} = t$}\\
\end{align*}

Hence we have a new random variable $Y$ whose cumulative density function is the same as that of $X$ and uses a single random sample by just plugging the single random sample $R$ in the function given, which is nothing but the inverse probability distribution of the random variable $X$.
\end{proof}

We have shown that a random weighted edge can be sampled using the rank function. However, we have considered a random sampling of a real value in the interval $[0,1]$ which is not possible using a computer, so we need to discretize the space and sample a discrete value in $[0,1]$ which results in a sampling error.

We show that dividing the space into very small discrete parts $\frac{1}{\poly(n)}$ would bound the error in the rank with high probability. Since the error in the rank of the edges is small, the probability with which order of ranks of any two edges would change due to this error is small too. Bounding the probability of error over all the pairs of edges we prove that the probability of order of rank of edges changing due to discretizing the space $[0,1]$ is very small.

Let us first discretize the space $[0,1]$ into $n^{10}$ parts and sample a number in $t\in [0,n^{10}]$ and return $\frac{t}{n^{10}}$. Recall that we sample a uniform random value $x \in [0,1]$ and use the function $r(e) = 1-(1-x)^{1/w(e)}$. The error in the sample $x$ due to discretizing the space is at most $\frac{1}{n^{10}}$. Hence, $dx = \frac{1}{n^{10}}$.

\begin{align*}
    \frac{d r(e)}{dx} &= \frac{1}{w(e)}\cdot(1-x)^{\frac{1}{w(e)}-1}\\
    &= \frac{1-r(e)}{w(e)(1-x)}\\
    &\leq \frac{1}{1-x} \tag{$w(e)\geq 1$ and $r(e)\in [0,1]$}\\
    &\leq n^3 \tag{when $x\leq 1-\frac{1}{n^3}$}\\
    d r(e) &\leq \frac{1}{n^7}\\
\end{align*}

To bound the error in $r(e)$ we ignore some range of the $x$, the probability of we coming up with such an $x$ is at most $\frac{1}{n^3}$. This has to be true for the uniform samples chosen for all the edges which can be at most $O(n^2)$. Hence union bound over all those edges the probability of error in $r(e)$ exceeding over $\frac{1}{n^7}$ is at most $\frac{m}{n^3}\leq \frac{1}{n}$. Hence with high probability we can assume that error in $r(e)$ is very small due to discretizing the space. 

We are not yet done, although the error in rank is small we have to bound the probability of the event in which the order of the edges according to rank is changed due to the error in the rank of the edges. We show that it occurs with very low probability. The high level idea is, because the error in rank is small for another edge to change the order it has to fall in the range of the small error and it occurs with very low probability.

Let $e,f$ are two edges and we have proved above that error in $r(e)$ is at most $\frac{1}{n^7}$. For $f$ to change order with respect to $e$, $r(f)$ has to lie in a range of $\frac{1}{n^7}$ of $r(e)$. Hence we have to bound the size of the domain for which the range is small. We have $d r(f)\leq \frac{c}{n^7}$ for some constant $c$.

\begin{align*}
    \frac{dr(f)}{dx} &= \frac{1}{w(f)}\cdot (1-x)^{\frac{1}{w(f)}-1}\\
    dx &= w(f)\cdot(1-x)^{1-\frac{1}{w(f)}}\cdot dr(f)\\
    &\leq \frac{c}{n^5} \tag{$1\leq w(f)\leq n^2$}\\
\end{align*}

Applying a union bound over all possible pairs of edges we have at most $O(n^4)$ pairs and thus error probability is at most $O(n^4)\cdot\frac{c}{n^5} \leq \frac{c}{n}$. Hence with high probability no pair of edges change their order of rank due to discretizing the space.

So the total error probability of the event in which the error in rank is high or any pair of edges swap their order due to the error in rank is at most $O(\frac{1}{n})$. Thus we say that order of edges with respect to rank does not change due to discretizing space with high probability. Since all we care about is the order of the edges with respect to rank while choosing an edge to contract, the process outputs a random edge with high probability.

\section{Local algorithm for minimal extreme set}\label{sec:minextset}

From \cite{cen2022edge} we know that all extreme sets in the graph form a laminar family. So all extreme sets containing a particular vertex say $x$ are contained in one another. So we define a minimal $k$-extreme set containing a particular vertex $x$ as follows.

\begin{definition}[Minimal $(x,k)$-extreme set]
A $k$-extreme set $S$ containing the vertex $x$ is called \emph{minimal $(x,k)$-extreme set} if no strict subset of $S$ containing $x$ is a $k$-extreme set.
\end{definition}

Note that this is guaranteed to exist if the graph $G$ is not $k$-connected i.e., the minimum cut of $G$ is $<k$.

We can define \emph{minimal $(x,\nu,\sigma,k)$-extreme set} analogously. However, this may not be guaranteed to exist as the volume of the $k$-extreme set need not be $<\nu$ or the cardinality need not be $< \sigma$.

Before we present an algorithm for finding the minimal $(x,\nu,k)$-extreme set in case if it exists, we need to observe some properties of the extreme sets that help us to prove the correctness of the algorithm. From \cite{cen2022edge} we know that the extreme sets in a graph form a laminar family. Hence all extreme sets containing $x$ are contained in one another. Out of all possible extreme sets containing $x$ let us consider those extreme sets that are $(\nu,\sigma,k)$-extreme sets i.e. they satisfy the volume, cardinality, and cut size constraints. If there exists at least one such extreme set containing $x$, then there exists a rank function respected by that $(x,\nu,\sigma,k)$-extreme set among the $O(\sigma^2\log n)$ random rank functions. Hence we can find the set using an algorithm similar to \Cref{alg:blp}.

Cactus graph is a graph in which every edge of the graph belongs to at most one simple cycle. From \cite{dinic1976system} we know that all minimum cuts of a graph $G$ can be represented in the form of a cactus graph $H$ where each vertex of the cactus maps to a disjoint set of vertices in the original graph $G$ and can also map to an empty set. Every cut induced in the cactus graph by the removal of a single non-cycle edge or two cycle edges that are part of same cycle will induce a corresponding cut in the original graph. Each such cut is minimum cut in the original graph. Cactus graph is succinct representation of all minimum cuts of the original graph. When we have a unique minimum cut in the graph, the cactus representation of minimum cuts of such a graph would be a single edge. From \cite{panigrahi2009near} we know that cactus representation of the graph can be found in near linear time. Hence we can check if a graph is having a unique minimum cut or not in near linear time. Let $\textsc{isExtreme}$ be the sub-routine to check if a graph is extreme set or not.

We formally define the problem as follows.

\begin{problem}\label{prob:exactminimal}
Given a graph $G = (V,E,w)$, parameters $\nu,\sigma,k$ and a vertex $x\in V$, find the $(x,\nu,\sigma,k)$-minimal extreme set or return $\bot$ if no such set exists.
\end{problem}

We will now present the local algorithm that solves \Cref{prob:exactminimal}. The algorithm is similar to \Cref{alg:blp} except that we also check if the returned set is extreme or not.

\begin{algorithm}[H]
\KwData{$G = (V,E,w),x,\nu,\sigma,k$} 
\KwResult{$(x,\nu,\sigma,k)$-minimal extreme set or $\bot$ if no $(x,\nu,\sigma,k)$-extreme set exists}
$S = V$\;
\For{$ i \in [O(\sigma^2 \log n)]$}
{
    $X = \{x\}$\;
    \While{$vol(X) < \nu$ and $|X| < \sigma$}
    {
        \If{$w(\delta(X)) < k$ and $\textsc{isExtreme}(G/\overline{X})$ and $X \subset S$}
        {
            $S = X$\;
            Break the while loop\;
        }
        Find the edge with minimum rank $r^i$, $e = (u,v) \in E(X,V\setminus X) = \delta(X)$\;
        Update $X = X \cup \{v\}$\;
    }
}
\eIf{$S \neq V$}
{
    \Return $S$\;
}
{
    \Return $\bot$\;
}
\caption{\textsc{MinimalExtremeSet}}
\label{alg:exactminimal}
\end{algorithm}

\begin{lemma}
\Cref{alg:exactminimal} solves the \Cref{prob:exactminimal} with high probability.
\end{lemma}

\begin{proof}
If there is no $(x,\nu,\sigma,k)$-extreme set then the algorithm would not update $S$ from $V$, we return $\bot$ at the end.

If there exists an $(x,\nu,\sigma,k)$-extreme set $S$ then as proved in \Cref{lem:rrrprob}, $S$ would respect at least one of the $O(\sigma^2\log n)$ random rank functions, say $r^*$. Thus set $S$ would be one among the $\sigma-1$ sets obtained using random contractions with respect to $r^*$. Since any other $(x,\nu,\sigma,k)$-extreme set strictly contains $S$ as extreme sets form a laminar family. We would find $S$ and return it without being replaced by any other extreme set.
\end{proof}

\begin{lemma}
\Cref{alg:exactminimal} runs in at most $\Tilde{O}(\nu\sigma^3)$ time.
\end{lemma}

\begin{proof}
For each of the $O(\sigma^2 \log n)$ rank functions, we perform $\sigma$ many edge contractions and each time we check whether the set is extreme or not, which takes near linear in the number of edges present in $G/\overline{X}$ which is at most $\Otil(\nu)$. It also takes $O(\nu\log \nu)$ for finding the minimum rank edges. So in total it takes at most $\Otil(\nu\sigma) + O(\nu\log \nu) = \Otil(\nu\sigma)$ per rank function. Since we do not know which rank function would be respected by the minimal extreme set, we have to run for every rank function making the running time at most $\Otil(\nu\sigma^3)$.
\end{proof}
\section{Uniqueness of maximal \texorpdfstring{$k$}{k}-edge-connected partition}\label{sec:mcpunique}

The goal of this section is to mainly prove the \Cref{prop:mcpunique}.

\begin{claim}\label{lem:mcpunique}
A maximal $k$-edge-connected subgraphs $\{V_{1},\dots,V_{z}\}$ of any graph $G$ is unique and form a partition of $V$. 
\end{claim}

\begin{proof}
First we prove that the maximal $k$-edge-connected subgraphs form a partition. Any maximal $k$-edge-connected subgraph $H$ is always an \emph{induced subgraph} of $G$. That is, $H = G[S]$ for some $S \subseteq V$. Otherwise, there is an edge $e \in G[S] \setminus E(H)$ and $H \cup {e}$ is a strict supergraph of $H$ that is $k$-edge-connected, contradicting the maximality of $H$. 

So if two maximal $k$-edge-connected graphs $G[A],G[B]$ have intersection i.e., $A\cap B \neq \emptyset$. Since both $G[A]$ and $G[B]$ are $k$-edge-connected hence $G[A \cup B]$ is also $k$-edge-connected as $A\cap B \neq \emptyset$. This contradicts that $G[A]$ is maximal $k$-edge-connected. Hence the sets are disjoint. Each vertex by itself is $k$-edge-connected hence, should belong to a maximal $k$-edge-connected subgraph. Thus they form a partition.

Now we prove that it is unique. Assume for contradiction that there exists at least two such partitions. Then there exists at least one vertex $v$ such that it belongs to different sets in the two maximal connected partitions. Again let $A,B$ be the two sets from two different partitions containing $v$ and both $G[A],G[B]$ are $k$-edge-connected. Using same argument as above we prove that the partition is unique.
\end{proof}

\section{Counter example for edge strength estimation} \label{counterexample}
Benczur and Karger
\cite{benczur2002randomized} gave an algorithm that provides estimates $\Tilde{k}_e$ for the strength of edges $k_e$ such that $\Tilde{k}_e \leq k_e$ for all $e \in E$ and $\sum_{e\in E}\frac{1}{\Tilde{k}_e} = O(n)$. These estimates are obtained by running $\textsc{Estimation}(G,1)$ sub-routine described in Lemma 4.9 of \cite{benczur2002randomized}. To get a $c$-approximation for edge strengths, we need that $\Tilde{k}_e \geq \frac{k_e}{c}$. However, the second condition given above is different and does not necessarily provide a good approximation guarantee. Below we give a counter example where these estimates are off from the actual values by a large factor.

Consider the following counter-example: a lollipop graph $G$ that has a path of length $n-n^{1/3}$ connected to a clique of size $n^{1/3}$. Note that $G$ contains less than $2(n-1)$ edges. Thus, the procedure $\textsc{Estimation}(G,1)$ will assign the edge-strength estimate of all edges to $1$. However, all edges from the clique have strength at least $n^{1/3}$. So the estimate of edge strength for the edges inside the clique is off by a $n^{1/3}$ factor.

\bibliographystyle{alpha}
\bibliography{ref}
\end{document}